\pdfoutput=1
\documentclass[11pt]{article}

\usepackage[utf8]{inputenc}
\IfFileExists{lmodern.sty}{\usepackage[T1]{fontenc}\usepackage{lmodern}}{}
\usepackage{amsmath,amssymb,amsthm,mathtools}
\usepackage{booktabs}
\usepackage{microtype}
\usepackage[margin=1.05in]{geometry}
\usepackage{enumitem}
\usepackage{xcolor}
\usepackage{tikz}
\usetikzlibrary{arrows.meta,calc,positioning,decorations.pathreplacing}
\usepackage{pgfplots}
\pgfplotsset{compat=1.18}
\usepgfplotslibrary{fillbetween,groupplots}
\usepackage[most]{tcolorbox}
\usepackage[colorlinks=true,linkcolor=blue!55!black,citecolor=blue!55!black,urlcolor=blue!55!black]{hyperref}
\usepackage[nameinlink,capitalise]{cleveref}

\newtheorem{theorem}{Theorem}[section]
\newtheorem{proposition}[theorem]{Proposition}
\newtheorem{lemma}[theorem]{Lemma}
\newtheorem{corollary}[theorem]{Corollary}
\newtheorem{mainthm}{Theorem}

\newtheorem*{pentcor}{Corollary (the Duan--Winter pentagon conjecture)}
\theoremstyle{definition}
\newtheorem{definition}[theorem]{Definition}
\newtheorem{remark}[theorem]{Remark}
\crefname{mainthm}{Theorem}{Theorems}

\DeclareMathOperator{\Tr}{Tr}
\DeclareMathOperator{\supp}{supp}
\DeclareMathOperator{\ran}{ran}
\DeclareMathOperator{\vect}{vec}
\DeclareMathOperator{\rk}{rank}

\newcommand{\vt}{\vartheta}
\newcommand{\cH}{\mathcal H}
\newcommand{\cG}{\mathcal G}
\newcommand{\Cmin}{C_{\min}}
\newcommand{\SNS}{S_{0,\mathrm{NS}}}
\newcommand{\Xs}{X(\cG,\le_*)}
\newcommand{\boxt}{\boxtimes}
\newcommand{\eps}{\varepsilon}
\newcommand{\ket}[1]{\lvert #1\rangle}
\newcommand{\bra}[1]{\langle #1\rvert}
\newcommand{\proj}[1]{\ket{#1}\!\bra{#1}}
\newcommand{\nrm}[1]{\lVert #1\rVert}

\title{The Lov\'asz number is not alone:\\[3pt]
\large a continuum of R\'enyi points in the entanglement-assisted spectrum of graphs}
\author{Nidhal Mghirbi \and Seth Douglas}
\date{October 2026}
\hypersetup{pdftitle={The Lovasz number is not alone: a continuum of Renyi points in the entanglement-assisted spectrum of graphs},pdfauthor={Nidhal Mghirbi and Seth Douglas},pdfsubject={October 2026}}

\begin{document}
\maketitle

\begin{abstract}
Li and Zuiddam's entanglement-assisted asymptotic spectrum of graphs consists
of the normalized graph parameters that are additive under disjoint unions,
multiplicative under strong products and monotone under
entanglement-assisted cohomomorphisms; the entanglement-assisted Shannon
capacity is their pointwise minimum.  To our knowledge, the Lov\'asz number
$\vt$ was its only known element.  We find a continuum of others.

Three quantities of Duan and Winter, $\log_2\vt$, the minimal Holevo capacity
$\Cmin$ and the logarithm of the no-signalling simulation cost $\Sigma$, are
the smallest information radii of classical--quantum channels compatible with
a graph, measured with Petz--R\'enyi divergences of orders $0$, $1$ and $2$.
At order zero this is due to Schrijver (as credited by Dalai) and to Duan and
Winter; at order one it is the min--max formula for the Holevo capacity.
Interpolating in the order gives a family $\vt_\alpha$, $0\le\alpha\le2$, and
we prove that every member lies in the spectrum.  The main step,
multiplicativity under strong products for channels that need not factor,
follows from an exact nesting identity for compressions.  Sandwiched R\'enyi
radii of order $\beta\in[1/2,\infty]$ give the same family, at
$\alpha=2-1/\beta$.

On the five-cycle the family takes every value in $[\sqrt5,5/2]$.  A
Hilbert--Schmidt inequality for the pentagon gives $\Cmin(C_5)=\log_2(5/2)$,
proving a conjecture of Duan and Winter.  We also settle the two open links of
their inequality chain: the graph-level no-signalling simulation rate always
equals its single-letter value $\log_2\Sigma(G)$, and a 32-vertex graph from
the Mermin--Peres magic square has $\Cmin<\log_2\Sigma$, certified in exact
rational arithmetic.  Finally, $\Sigma$ never exceeds the projective rank of
the complement; a rank inequality for the same square makes this strict on a
1440-vertex graph.
\end{abstract}

\section{Introduction}\label{sec:intro}

Shannon's zero-error capacity $\Theta(G)$ of a graph is notoriously hard to
compute \cite{Shannon1956}.  Zuiddam found a clean way to organize what we
know about it \cite{Zuiddam2019}.  Consider every graph parameter that adds
under disjoint union, multiplies under the strong product, equals $1$ on a
single vertex and respects the cohomomorphism preorder.  These parameters form
the \emph{asymptotic spectrum of graphs}, and Strassen's duality theorem
\cite{Strassen1988} says that $\Theta$ is their pointwise minimum.  So every
spectral point is an upper bound on the Shannon capacity, and many of the
classical upper bounds turn out to be spectral points \cite{Zuiddam2019}: the
Lov\'asz number \cite{Lovasz1979}, the fractional Haemers bounds
\cite{Haemers1979,Blasiak2013,BukhCox2019}, the complement of the projective
rank \cite{CMRSSW2014,MancinskaRoberson2016}, and the fractional clique cover
number.  The ordinary spectrum is by now known to be uncountable
\cite{Vrana2021,Vrana2022}.

Li and Zuiddam carried the theory over to quantum and entanglement-assisted
communication \cite{LiZuiddam2020}.  With shared entanglement the preorder
changes.  Now $G\le_*H$ means that, with the help of an entangled state, one
use of a channel with confusability graph $H$ can do the job of one use of a
channel with confusability graph $G$.  The duality survives: the
entanglement-assisted Shannon capacity $\Theta_*$ is the minimum over a new
spectrum $\Xs$, where $\cG$ denotes the class of finite simple graphs.  This
spectrum has proved much harder to populate.  Beigi and,
independently, Duan, Severini and Winter had shown that $\Theta_*\le\vt$
\cite{Beigi2010,DSW2013}, and Li and Zuiddam proved that $\vt$ itself lies in
$\Xs$, using the entanglement-assisted monotonicity of $\vt$ due to Cubitt et
al.~\cite{CMRSSW2014}.  They found no other element.  They also showed
that the fractional clique cover number is not in the spectrum
$X(\cG,\le_q)$ of the quantum preorder \cite[Theorem~22]{LiZuiddam2020}
(Theorem~21 in arXiv:1810.00744v2).  This point of the ordinary spectrum
equals, by linear programming duality, the fractional packing number
$\alpha^*$ used below, and $X(\cG,\le_q)$ contains $\Xs$, so it is not in
$\Xs$ either.  And they
suggested that the conjecture $\Theta_*=\vt$ might be proved by showing that
$\vt$ is the least element of the spectrum, ``or even the only point''
\cite{LiZuiddam2020}.

It is not the only point.  This paper constructs a continuum of others.

\subsection{Three quantities, one family}

We start from a paper of Duan and Winter \cite{DuanWinter2016} on zero-error
communication assisted by no-signalling correlations.  For a confusability
graph $G$ they compare several numbers attached to the classical--quantum
channels $x\mapsto\rho_x$ whose confusability graph is $G$.  There is the
Lov\'asz number $\vt(G)$; Duan and Winter show that $\log_2\vt(G)$ is the
smallest no-signalling-assisted zero-error capacity of a channel compatible
with $G$ \cite[Theorem~5]{DuanWinter2016}.  There is $\Cmin(G)$, the
infimum of the Holevo capacities of such channels, and $\Sigma(G)$, the infimum over
such channels of the semidefinite-programming value that governs their exact
simulation with no-signalling assistance.  They prove the chain
\begin{equation}\label{eq:dw-chain}
 \log_2\vt(G)\ \le\ \Cmin(G)\ \le\ \SNS(G)\ \le\ \log_2\Sigma(G)\ \le\
 \log_2\alpha^*(G),
\end{equation}
where $\SNS$ is an asymptotic simulation cost and $\alpha^*$ is Shannon's
fractional packing number \cite{Shannon1956} (\cref{sec:dw-quantities} has
the definitions).  By results of Cubitt, Leung, Matthews and Winter
\cite{CLMW2011}, among classical channels with confusability graph $G$ the
smallest asymptotic no-signalling-assisted zero-error capacity and the
smallest asymptotic no-signalling-assisted simulation cost both equal
$\log_2\alpha^*(G)$.  This is how $\alpha^*$ enters the chain
\cite[Section~V]{DuanWinter2016}.

Three of the quantities in \eqref{eq:dw-chain}, $\log_2\vt$, $\Cmin$ and
$\log_2\Sigma$, are really one quantity in disguise.  Fix a channel compatible with $G$ and a state $\sigma$.  Measure
the distance from each output $\rho_x$ to $\sigma$ with the Petz--R\'enyi
divergence $D_\alpha$ and take the worst case.  This worst case,
$\max_xD_\alpha(\rho_x\Vert\sigma)$, is an order-$\alpha$ \emph{information
radius}.  Minimizing it over channels and centers gives a graph parameter:
\begin{equation}\label{eq:def-intro}
 \log_2\vt_\alpha(G)
 =\inf_{(\rho_x),\,\sigma}\ \max_{x\in V(G)}D_\alpha(\rho_x\Vert\sigma),
 \qquad 0\le\alpha\le2.
\end{equation}
Then
\begin{equation}\label{eq:endpoints-intro}
 \vt_0=\vt,\qquad \vt_1=2^{\Cmin},\qquad \vt_2=\Sigma.
\end{equation}
The first is a theorem of Dalai, who credits it to Schrijver
\cite[Theorem~8]{Dalai2013TIT}.  Duan and Winter prove it in the proof of
their Theorem~5, by purifying the center \cite[Section~V]{DuanWinter2016},
and the case of a pure center with general projections is due to Cabello,
Severini and Winter \cite{CSW2010}; Dalai and Winter note the observations of
Schrijver and of Duan and Winter \cite{DalaiWinter2017}.  The middle identity
is the familiar min--max formula for the Holevo capacity
\cite{OhyaPetzWatanabe1997,SchumacherWestmoreland2001}; the third is
elementary once the right formula is in place
(\cref{prop:endpoints}).  So the three Duan--Winter quantities are samples,
at orders $0$, $1$ and $2$, of one monotone family
$\alpha\mapsto\vt_\alpha(G)$, which is nondecreasing on $[0,2]$ and continuous
on $[0,1)$ (\cref{prop:endpoints,prop:continuity}).  This also explains why
they come in the order they do, since $D_\alpha$ is nondecreasing in $\alpha$.

What happens in between?

\subsection{Results}

\begin{mainthm}[every order up to two is a spectral point]\label{thm:A}
For every $\alpha\in[0,2]$, the graph parameter $\vt_\alpha$ belongs to the
entanglement-assisted asymptotic spectrum $\Xs$.
\end{mainthm}

For $\alpha=0$, where $\vt_0=\vt$, membership is due to Li and Zuiddam
\cite{LiZuiddam2020}.  The cases $\alpha>0$ are new.

Normalization and additivity are routine.  Monotonicity under $\le_*$ is a
data-processing argument, and it is why our argument stops at order two:
above order two, Petz divergences no longer satisfy data processing
(\cref{rem:above-two}).  The real work is
multiplicativity, $\vt_\alpha(G\boxt H)=\vt_\alpha(G)\vt_\alpha(H)$.  One
inequality follows by tensoring witnesses.  The other has to control
representations of $G\boxt H$ that are \emph{not} tensor products, and there
is no reason to expect a near-optimal representation of a product to factor.
\Cref{sec:spectral} handles this with an exact nesting identity for the
compressions that appear in the projection formula~\eqref{eq:projection-formula}.

Are the new points actually new?  The pentagon answers this.

\begin{mainthm}[the pentagon]\label{thm:B}
For $\tfrac12\le\alpha\le2$ we have $\vt_\alpha(C_5)=\tfrac52$.  For
$0\le\alpha\le\tfrac12$, writing $p=1/(1-\alpha)\in[1,2]$,
\begin{equation}\label{eq:pentagon-intro}
 \sqrt5\left(\frac{\sqrt5}{2}\right)^{p-1}
 \ \le\ \vt_\alpha(C_5)\ \le\
 \min\left\{\frac52,\ \sqrt5\cdot5^{(p-1)/2}\right\}.
\end{equation}
The map $\alpha\mapsto\vt_\alpha(C_5)$ is continuous and carries
$[0,\tfrac12]$ onto $[\sqrt5,\tfrac52]$.  In particular $\Xs$ contains
continuum many distinct points of the form $\vt_\alpha$.
\end{mainthm}

\Cref{fig:pentagon-curve} shows what we know about the pentagon curve.
Since $\vt_1(C_5)=\vt_2(C_5)=5/2$, \cref{thm:B} gives the values of $\Cmin$
and $\Sigma$ on the pentagon.  That is precisely what Duan and Winter
conjectured after their Hadamard-graph example, two paragraphs after the
chain~\eqref{eq:dw-chain}, where they note that a rigorous
proof ``has so far eluded'' them \cite[Eq.~(54), which is (53) in
arXiv:1409.3426v3, and the discussion after it]{DuanWinter2016}.

\begin{pentcor}
$\log_2\sqrt5<\Cmin(C_5)=\SNS(C_5)=\log_2\Sigma(C_5)=\log_2\alpha^*(C_5)
=\log_2\tfrac52.$
\end{pentcor}

The values $\vt(C_5)=\sqrt5$ \cite{Lovasz1979} and $\alpha^*(C_5)=\frac52$
\cite{Shannon1956} are classical, and the upper bounds on $\Cmin(C_5)$,
$\SNS(C_5)$ and $\log_2\Sigma(C_5)$ are Duan and Winter's
chain~\eqref{eq:dw-chain}.  What is new is the lower bound
$\Cmin(C_5)\ge\log_2\frac52$.

\begin{figure}[t]
\centering
\begin{tikzpicture}
\begin{groupplot}[
  group style={group size=2 by 1, horizontal sep=0pt, yticklabels at=edge left},
  scale only axis, height=5.0cm,
  ymin=2.200, ymax=2.548,
  ytick={2.2360679775,2.3,2.4,2.5},
  yticklabels={$\sqrt5$,$2.3$,$2.4$,$\tfrac52$},
  tick label style={font=\small},
  set layers,
  clip=false,
]
\nextgroupplot[width=7.4cm, xmin=0, xmax=0.5,
  xtick={0,0.1,0.2,0.3,0.4,0.5},
  xticklabels={$0$,$0.1$,$0.2$,$0.3$,$0.4$,$\tfrac12$},
  ylabel={$\vt_\alpha(C_5)$}, ylabel style={font=\small}]
  \addplot[name path=up, draw=none, domain=0:0.5, samples=240]
     {min(2.5, 5^(1/(2-2*x)))};
  \addplot[name path=lo, draw=none, domain=0:0.5, samples=240]
     {sqrt(5)*(sqrt(5)/2)^(x/(1-x))};
  \addplot[gray!22, on layer=pre main] fill between[of=up and lo];
  \addplot[very thick, orange!85!black, domain=0:0.1217646, samples=200]
     {5^(1/(2-2*x))};
  \addplot[thick, black, densely dashed, domain=0.1217646:0.5] {2.5};
  \addplot[very thick, blue!70!black, domain=0:0.5, samples=240]
     {sqrt(5)*(sqrt(5)/2)^(x/(1-x))};
  \draw[densely dotted, gray] (axis cs:0.1217646,2.200) -- (axis cs:0.1217646,2.5);
  \node[font=\scriptsize, gray!40!black, anchor=south west]
     at (axis cs:0.128,2.203) {$\alpha\approx0.122$};
  \addplot[only marks, mark=*, mark size=2pt] coordinates {(0,2.2360679775) (0.5,2.5)};
  \addplot[only marks, mark=x, mark size=2.4pt, thick, red!75!black] coordinates
     {(0.04,2.312314) (0.06,2.353826) (0.08,2.397128) (0.09,2.419036)
      (0.1,2.440835) (0.11,2.462321) (0.12,2.483319) (0.125,2.493589)};
  \node[font=\scriptsize, anchor=north west] at (axis cs:0.006,2.233) {$\vt(C_5)=\sqrt5$};
  \node[font=\scriptsize, orange!70!black, anchor=west]
     at (axis cs:0.004,2.47) {umbrella};
  \node[font=\scriptsize, blue!60!black, anchor=north west]
     at (axis cs:0.27,2.305) {interpolation bound};
  \node[font=\scriptsize, gray!30!black, align=center]
     at (axis cs:0.24,2.445) {the curve lies\\in the shaded band};
\nextgroupplot[width=3.3cm, xmin=0.5, xmax=2.08,
  xtick={1,2}, xticklabels={$1$,$2$}]
  \addplot[very thick, black, domain=0.5:2] {2.5};
  \addplot[only marks, mark=*, mark size=2pt] coordinates {(1,2.5) (2,2.5)};
  \node[font=\scriptsize, anchor=south] at (axis cs:1,2.505) {$2^{\Cmin}$};
  \node[font=\scriptsize, anchor=south] at (axis cs:2,2.505) {$\Sigma$};
  \node[font=\scriptsize, align=center] at (axis cs:1.29,2.36)
     {flat at $\tfrac52$\\(\cref{thm:pentagon-half})};
\end{groupplot}
\node[font=\small] at ($(group c1r1.south east)+(0,-0.75)$) {order $\alpha$};
\end{tikzpicture}
\caption{Bounds on $\alpha\mapsto\vt_\alpha(C_5)$.  The
horizontal scale changes at $\alpha=\tfrac12$.  On $[\tfrac12,2]$ the curve is
exactly $\tfrac52$, which gives $\Cmin(C_5)=\log_2\tfrac52$ and
$\Sigma(C_5)=\tfrac52$.  On $[0,\tfrac12]$ it is continuous and nondecreasing
and lies between the interpolation lower bound (blue) and the minimum of
$\tfrac52$ and the Lov\'asz-umbrella bound obtained from the limiting
handle center (orange); where this minimum is $\tfrac52$ it is drawn dashed,
since there it is only an upper bound.  The umbrella bound is better than $\tfrac52$
only for $\alpha<1-\ln5/(2\ln\frac52)\approx0.122$; the crossing of these
upper bounds does not by itself identify a transition in the exact graph value.
Crosses: numerical upper bounds from explicit witnesses, the umbrella vectors
of \cref{fig:pentagon-geometries}(b) with the faithful center
$(1-x)\proj c+\tfrac x2(I-\proj c)$ and $x$ optimized at each order
(verification script \texttt{pentagon\_numerical\_bounds.py}; see Data and
code availability).  They improve on the
limiting handle center only slightly, for example $2.44084<2.44521$ at
$\alpha=0.1$, and they stay below $\tfrac52$ until $\alpha\approx0.128$.
Continuity and the endpoint values $\sqrt5$ and $\tfrac52$ force the curve
through every value in between.}
\label{fig:pentagon-curve}
\end{figure}
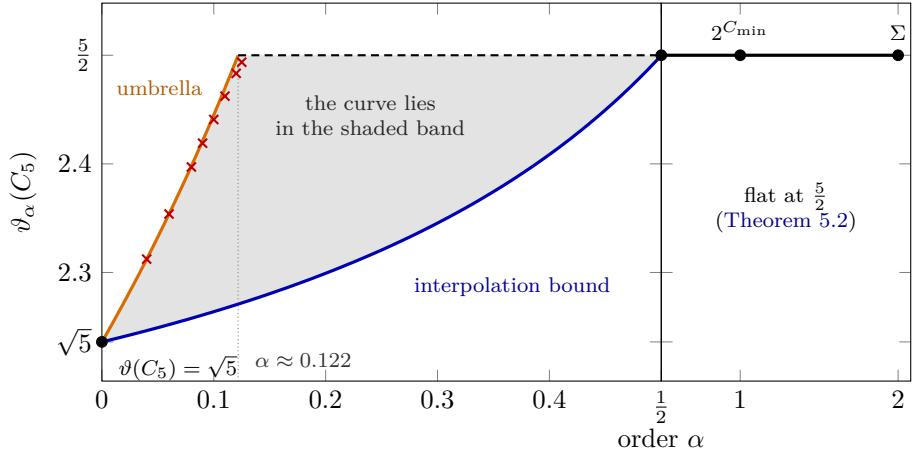

The pentagon cannot tell order one from order two.  A different graph can.

\begin{mainthm}[order one is not order two]\label{thm:C}
There is a graph $G_{32}$ on $32$ vertices, built from the Mermin--Peres
magic square, such that
\[
 2^{\Cmin(G_{32})}\ \le\ 6.186983\ldots\ <\ 6.194232\ldots\ \le\ \Sigma(G_{32}).
\]
Both outer bounds are certified by exact rational data.  In particular
$\vt_1\ne\vt_2$.
\end{mainthm}

In exact form, the two certified bounds are
\[
 2^{\Cmin(G_{32})}\ \le\ \frac{32329248967670751}{5225365592240894}
 \ <\ \frac{93824992214030}{15147154828811}\ \le\ \Sigma(G_{32}),
\]
the reciprocals of the rationals in~\eqref{eq:g32-rationals}.
For this graph $\vt(G_{32})=6$ and $\alpha^*(G_{32})=\tfrac{32}5$, so both
bounds lie well inside the interval $[\vt,\alpha^*]$ of the Duan--Winter chain
(\cref{sec:small}).

Order two also has a natural upper bound of its own.  Every
$d/r$-representation of the complementary graph $\overline G$, in the sense of
Man\v{c}inska and Roberson \cite{MancinskaRoberson2016}, is a flat witness for
$G$, so $\Sigma(G)\le\xi_f(\overline G)$, where $\xi_f$ is the projective rank
(\cref{sec:projective}).  The graph parameter $G\mapsto\xi_f(\overline G)$
lies in Li and Zuiddam's spectrum for the quantum preorder, and whether it
lies in $\Xs$ is open \cite{LiZuiddam2020}.  The two sides agree on the
pentagon and on Duan and Winter's Hadamard example, and if they agreed on
every graph, \cref{thm:A} would settle that question.  They do not.

\begin{mainthm}[order two is not the projective rank]\label{thm:D}
There is a graph $G_{1440}$ on $1440$ vertices, built from the same magic
square, such that
\[
 \Sigma(G_{1440})\ \le\ \frac{215940}{899}=240.2002\ldots\ <\
 240.2033\ldots=\frac{14172}{59}\ \le\ \xi_f\bigl(\overline{G_{1440}}\bigr).
\]
\end{mainthm}

The upper bound is an explicit order-two witness: the magic-square
representation with a center tilted slightly toward one of its vectors.  The lower bound is a rank inequality for the magic
square that holds in every dimension (\cref{lem:rank-ineq}); both bounds are
proved by hand.

\subsection{A taste of the method: the pentagon}
\label{sec:taste}

Given the projection formula of \cref{sec:compressions}, the pentagon
conjecture has a short proof.  Label
$C_5$ by $\mathbb Z_5$ with $i\sim i\pm1$.  In a representation of the
pentagon, non-neighbours get orthogonal projections: $P_iP_{i+2}=0$ for all
$i$.  At order $\alpha=\frac12$, the formula of \cref{sec:compressions} for
$\vt_\alpha$ becomes
\[
 \frac{1}{\vt_{1/2}(G)}
 =\sup_{(P_x),\,\sigma}\ \min_x\ \nrm{P_x\sigma^{1/2}P_x}_2^2 ,
\]
the supremum running over representations $(P_x)$ of $G$ and faithful states
$\sigma$.  Now the one inequality we need: for every representation of $C_5$
and every matrix $X$,
\begin{equation}\label{eq:hs-intro}
 \sum_{i\in\mathbb Z_5}\nrm{P_iXP_i}_2^2\ \le\ 2\,\nrm{X}_2^2 .
\end{equation}
The reason is local.  The two neighbours $i-1$ and $i+1$ of a vertex are not
adjacent to each other, so $P_{i-1}+P_{i+1}$ is again a projection.  Seen from
inside $P_i$, the sum of all five projections is therefore at most twice the
identity (\cref{lem:hs}).  Put $X=\sigma^{1/2}$.  The right side of~\eqref{eq:hs-intro}
is $2\Tr\sigma=2$, so some vertex has $\nrm{P_i\sigma^{1/2}P_i}_2^2\le2/5$,
and $\vt_{1/2}(C_5)\ge5/2$.  Since $\vt_\alpha$ is nondecreasing in $\alpha$
(\cref{prop:endpoints}(ii)), the
same bound holds for $\vt_1=2^{\Cmin}$ and for $\vt_2=\Sigma$.

The matching upper bound is the obvious classical channel: input $i$ produces a
uniformly random element of $\{i,i+1\}$.  Its Holevo capacity is
$\log_25-1=\log_2\frac52$, and $T=\frac12 I$ dominates every output with
$\Tr T=\frac52$.  Duan and Winter's chain~\eqref{eq:dw-chain} squeezes $\SNS$
in between, and $\alpha^*(C_5)=\frac52$ is classical \cite{Shannon1956}.
Besides \cref{lem:hs}, the argument uses the projection formula, the
identities $\vt_1=2^{\Cmin}$ and $\vt_2=\Sigma$ of \cref{prop:endpoints}(i)
(the first is the radius formula for the Holevo capacity) and monotonicity
in~$\alpha$.

\subsection{What is new}\label{sec:new}

\paragraph{New points in the entanglement-assisted spectrum.}
To our knowledge, the parameters $\vt_\alpha$ with $\alpha>0$ are the first
elements of $\Xs$ other than $\vt$.  Of the two routes to $\Theta_*=\vt$ that
Li and Zuiddam suggested \cite{LiZuiddam2020}, \cref{thm:A,thm:B} rule out the
stronger one, which was to show that $\vt$ is the only point.  They say
nothing about whether $\vt$ is the \emph{least} point.  Every $\vt_\alpha$
dominates $\vt$, so the bound
$\Theta_*\le\vt$ is not improved and the conjecture $\Theta_*=\vt$ is
untouched.  A continuum with the same values on $C_5$ was already available
in the \emph{ordinary} asymptotic spectrum.  Vrana's log-convex interpolation
\cite{Vrana2021} between $\vt$ and a fractional Haemers bound, whose value on
$C_5$ is $\frac52$ \cite[Proposition~4]{BukhCox2019}, evaluates the
vertex-transitive graph $C_5$ log-affinely \cite[Section~5]{Vrana2022}, so
its values on $C_5$ fill $[\sqrt5,\frac52]$.  What is new here is that the
points $\vt_\alpha$ lie in the entanglement-assisted spectrum $\Xs$.

\paragraph{Multiplicativity for joint representations.}
Duan and Winter proved that $\Sigma$ is multiplicative for two \emph{fixed}
cq-graphs \cite[Proposition~18]{DuanWinter2016}.  Applied to representations
of $G\boxt H$ that are tensor products, their result gives the order-two
identity of \cref{thm:product} for the infimum restricted to such
representations.  \Cref{thm:product} proves the graph-level identity without
that restriction, at every order in $[0,2]$; it does not recover their
statement for fixed cq-graphs.  At the two lower endpoints the graph-level
identity is known: order zero
is Lov\'asz's multiplicativity of $\vt$ \cite{Lovasz1979}, and order one is
Duan and Winter's additivity of $\Cmin$ under the strong product
\cite[Lemma~23]{DuanWinter2016}.  Every order in $(0,1)\cup(1,2]$ is new, and
order two is the one with an operational consequence.  There it gives a
single-letter formula,
\[
 \SNS(G)=\log_2\Sigma(G),
\]
for graph-level asymptotic no-signalling simulation: at each block length $n$
one may use any channel compatible with $G^{\boxt n}$, and such a
channel may correlate the $n$ uses.  Duan and Winter remarked that this
identity ``may be true'', and that proving it requires the additivity of
$\log\Sigma$ under the strong product, which they describe as unknown
\cite[Section~V]{DuanWinter2016}.  In their concluding section they list this
graph-level simulation cost as open \cite[Section~VII]{DuanWinter2016}.  For
a fixed cq-graph they had already shown that the asymptotic cost is
single-letter \cite[Proposition~18 and Theorem~19]{DuanWinter2016}.  At the
level of confusability graphs, \cref{thm:product} reduces asymptotic
no-signalling simulation to a one-shot problem.  Their general problem, for
non-commutative bipartite graphs that are not cq-graphs, is not addressed
here.

\paragraph{The Duan--Winter chain, link by link.}
With~\eqref{eq:endpoints-intro} in hand, every link of~\eqref{eq:dw-chain}
can now be decided (\cref{fig:chain}).  Duan and Winter showed that the last
link can be strict.  They conjectured that the first is strict on the
pentagon, which is the corollary above.  They left open whether either of the
two middle links can be strict, and remarked that the third might be an
equality.  It is always an equality, by multiplicativity at order two
(\cref{thm:product,rem:single-letter}).  The second
can be strict (\cref{thm:C}).  For a fixed cq-graph with two non-orthogonal
pure outputs, Duan and Winter had already shown the corresponding
fixed-channel inequalities to be strict \cite[Section~IV.D]{DuanWinter2016};
what is new is a separation of the graph-level quantities.  They also wrote
that ``it seems that
$\Sigma(G)$ and $2^{\Cmin(G)}$ are two new interesting, multiplicative graph
parameters, different from both $\vt(G)$ and $\alpha^*(G)$''
\cite[Section~VII]{DuanWinter2016}.  For $2^{\Cmin}$ multiplicativity is their
Lemma~23, and \cref{thm:product} gives it for $\Sigma$.  Their Hadamard-graph
example already separates both parameters from $\alpha^*$, and the pentagon
separates both from $\vt$.

\paragraph{One family instead of three numbers.}
Duan and Winter already wrote $\log_2\Sigma$ as a conditional min-entropy
\cite[Theorem~2]{DuanWinter2016}, which is a max-relative-entropy radius.
What seems to be new is that, once
one minimizes over states supported in a fixed projection, this radius
coincides with the \emph{Petz} radius of order two.  All three quantities then
sit in a single monotone family.  The sandwiched R\'enyi graph radii give
this entire family as well: for $\frac12\le\beta\le\infty$,
$\widetilde\vt_\beta=\vt_{2-1/\beta}$ (\cref{prop:sandwiched}).  Thus the
sandwiched family runs from $\vt$ at $\beta=\frac12$, through $2^{\Cmin}$
at $\beta=1$, to $\Sigma$ at $\beta=\infty$, and gives the same spectral
points.  The equality is between optimized graph radii, not divergences of
fixed pairs of states.  The picture has a useful by-product: every
$\vt_\alpha$ lies between $\vt$ and $\alpha^*$ (\cref{prop:endpoints}).  At
orders $0$, $1$ and $2$ these bounds follow from Duan and Winter's chain; at
other orders they follow from the chain together with monotonicity in
$\alpha$, and flat witnesses give the upper bound directly.  So the family
collapses to a single value on graphs
with $\vt=\alpha^*$, perfect graphs among them, and it can only separate
points on graphs with $\vt<\alpha^*$.

\paragraph{Order two is not the projective rank.}
Flat witnesses built from $d/r$-representations of the complement give
$\Sigma(G)\le\xi_f(\overline G)$, with equality whenever $\vt=\alpha^*$, on
$C_5$ and on Duan and Winter's Hadamard example.  To our knowledge, neither
this inequality nor the question of equality appears in earlier work.  Had
equality
held for every graph, \cref{thm:A} would have placed $\xi_f(\overline G)$ in
$\Xs$, which Li and Zuiddam leave open \cite[after Proposition~24]{LiZuiddam2020}.
\Cref{thm:D} closes that route, and since $\vt_\alpha\le\Sigma$ for
$\alpha\le2$ it does so for every member of the family; it does not decide
the question.  The
lower bound on $\xi_f$ uses no certificate.  It is one linear inequality
between the normalized ranks of a representation of the $24$-vertex graph
built from the magic square,
tight on the magic-square representation itself, and its proof is a robust
form of the anticommutation relations that every operator solution of the
magic square obeys \cite[Lemma~6]{Ji2013}: up to errors controlled by the
defects of the representation, cell observables in different rows and
columns anticommute, which forces their traces to be small.  Rank-counting
lower bounds on the projective rank have been used before, with spectral
input \cite{WocjanElphick2019}; we know of no lower bound of this
contextual kind.

\begin{figure}[t]
\centering
\begin{tikzpicture}[
  q/.style={font=\normalsize, inner sep=2pt, minimum height=0.6cm},
  ord/.style={font=\scriptsize, gray!50!black},
  st/.style={draw, rounded corners=2pt, align=center, font=\scriptsize,
             text width=2.8cm, inner sep=3pt, minimum height=1.3cm},
  new/.style={st, fill=teal!10, draw=teal!60!black},
  old/.style={st, fill=gray!8, draw=gray!55},
  val/.style={font=\small},
]
\def\dx{3.15}
\node[q] (a) at (0,0)      {$\log_2\vt$};
\node[q] (b) at (\dx,0)    {$\Cmin$};
\node[q] (c) at (2*\dx,0)  {$\SNS$};
\node[q] (d) at (3*\dx,0)  {$\log_2\Sigma$};
\node[q] (e) at (4*\dx,0)  {$\log_2\alpha^*$};
\foreach \u/\v in {a/b,b/c,c/d,d/e} \node at ($(\u)!0.5!(\v)$) {$\le$};
\node[ord] at (0,0.55)     {order $\alpha=0$};
\node[ord] at (\dx,0.55)   {order $\alpha=1$};
\node[ord] at (2*\dx,0.55) {asymptotic};
\node[ord] at (3*\dx,0.55) {order $\alpha=2$};
\node[ord] at (4*\dx,0.55) {LP bound};
\node[new] at ($(a)!0.5!(b)+(0,-1.25)$)
  {\textbf{strict on $C_5$}\\ conjectured by DW;\\ proved here, Thm~\ref{thm:B}};
\node[new] at ($(b)!0.5!(c)+(0,-1.25)$)
  {\textbf{can be strict}\\ open in DW;\\ $G_{32}$, Thm~\ref{thm:C}};
\node[new] at ($(c)!0.5!(d)+(0,-1.25)$)
  {\textbf{always equal}\\ open in DW;\\ Thm~\ref{thm:product} at $\alpha=2$};
\node[old] at ($(d)!0.5!(e)+(0,-1.25)$)
  {\textbf{can be strict}\\ shown by DW\\ \phantom{x}};
\node[font=\small, anchor=east] at (-0.55,-2.45) {$C_5$:};
\node[val] at (0,-2.45)      {$1.161$};
\node[val] at (\dx,-2.45)    {$1.322$};
\node[val] at (2*\dx,-2.45)  {$1.322$};
\node[val] at (3*\dx,-2.45)  {$1.322$};
\node[val] at (4*\dx,-2.45)  {$1.322$};
\node[font=\small, anchor=east] at (-0.55,-3.0) {$G_{32}$:};
\node[val, gray] at (0,-3.0)     {$\cdot$};
\node[val] at (\dx,-3.0)   {$\le2.6293$};
\node[val] at (2*\dx,-3.0) {$\ge2.6309$};
\node[val] at (3*\dx,-3.0) {$\ge2.6309$};
\node[val, gray] at (4*\dx,-3.0) {$\cdot$};
\end{tikzpicture}
\caption{The Duan--Winter chain~\eqref{eq:dw-chain} after this paper.  The
three quantities in the first four columns that carry an order are the members
$\alpha=0,1,2$ of one monotone family.  The last column is Shannon's fractional packing
(LP) bound \cite{Shannon1956}.  Shaded boxes record the status of each link;
teal marks what is new here.  The bottom rows give values in bits (dots are
values we do not need).}
\label{fig:chain}
\end{figure}

\subsection{Related work}

Dalai showed that the Lov\'asz bound and the sphere-packing bound for
classical--quantum channels are both information radii of Csisz\'ar type
\cite{Csiszar1995} built from R\'enyi divergences
\cite{Dalai2013,Dalai2013TIT}.  At the level
of graphs he considered the order-zero radius $\vt_{\mathrm{sp}}$, the
minimum over projector representations $(U_x)$ of $G$ and density operators
$F$ of $\max_x\log(1/\Tr U_xF)$, and proved that it equals $\log\vt(G)$, a
theorem he credits to Schrijver \cite[Theorem~8]{Dalai2013TIT}.  This is our
identity $\vt_0=\vt$.  Duan and Winter prove it as well, by purifying the
center in the proof of their Theorem~5 \cite[Section~V]{DuanWinter2016}; the
case of a pure center with general projections is due to Cabello, Severini
and Winter \cite{CSW2010}, and Dalai and Winter note the observations of
Schrijver and of Duan and Winter \cite{DalaiWinter2017}.  The family
$\vt_\alpha$ continues this graph-level order-zero radius to all orders
$0\le\alpha\le2$.  For a fixed cq-channel, Duan and Winter already express
their quantities as radii: the no-signalling-assisted zero-error capacity of
their Theorem~4 is the order-zero radius (by the minimax formula in their
proof of Theorem~5), the simulation cost of their Theorem~2 is given by a
conditional min-entropy, that is, by a max-relative-entropy radius, and the
Holevo capacity is the relative-entropy radius \cite{DuanWinter2016}.  What
is new here is the graph-level family at every order in $[0,2]$ and the
identification of its order-two member with a Petz radius.

Recent work concerns related channel quantities.  Non-signalling-assisted
simulation cannot increase the R\'enyi capacity of a classical channel, at
any order.  This follows from the data-processing inequality for channel
divergences under non-signalling maps of Oufkir, Yao and Berta
\cite[Lemma~24]{OYB2026} (Lemma~25 in arXiv:2509.15200v1), and Shariat,
Hoogsteder-Riera and Tomamichel give an alternative variational proof that
works directly with the information-radius form of the capacity
\cite{SHT2026}.  This is a channel-level classical analogue of our
monotonicity step (\cref{thm:ea-mono}).  Concurrently, Cheng and Berta show that the
Petz--R\'enyi channel information of quantum-input channels can be strictly
superadditive for every $\alpha\in[\frac12,1)$ \cite{ChengBerta2026}.  This
does not conflict with \cref{thm:product}, which concerns classical-input
families and covers arbitrary, possibly entangled, joint representations of
$G\boxt H$.  Petz radii of a fixed cq-channel also have an operational
meaning: Petz--R\'enyi capacities of Csisz\'ar type determine the error
exponent of coding assisted by activated no-signalling correlations
\cite{OTB2024}.

In the ordinary spectrum, continua arise in a different way.  Vrana's
probabilistic refinement makes the ordinary spectrum log-convex in a suitable
sense \cite{Vrana2021}, and so interpolating between two known points produces
new ones.  We do not know whether an analogue holds for $\le_*$.
\Cref{sec:open} gives a heuristic reason why, even if it does, $\vt_1$ should
not be an interpolation between $\vt$ and $\Sigma$.  The quantum preorder
$\le_q$ relates fewer pairs of graphs than $\le_*$ ($G\le_qH$ implies
$G\le_*H$), so monotonicity under $\le_q$ is the weaker requirement and
$X(\cG,\le_q)$ contains $\Xs$.  Li and Zuiddam place the fractional real and
complex Haemers bounds and the complement of the projective rank in
$X(\cG,\le_q)$ \cite[Theorem~22]{LiZuiddam2020}, and Gao, Gribling and Li
place a tracial Haemers bound in the asymptotic spectrum for the
commuting-operator version of the quantum preorder \cite{GaoGriblingLi2022}.
To our knowledge none of these is known to be monotone under $\le_*$; for
the Haemers bounds and the projective rank, Li and Zuiddam state this
explicitly \cite{LiZuiddam2020}.  Finally, Vrana studied which points of the ordinary spectrum
extend to non-commutative graphs \cite{Vrana2022}.  We stay with classical
graphs throughout.

\paragraph{Organization.}
\Cref{sec:setup} fixes conventions.  \Cref{sec:compressions} rewrites
$\vt_\alpha$ in terms of projections and proves the basic properties of the
family.  \Cref{sec:spectral} proves \cref{thm:A}.  \Cref{sec:pentagon}
treats the pentagon and the Duan--Winter chain, \cref{sec:small}
separates orders one and two, and \cref{sec:projective} separates order two
from the projective rank.  Open problems are collected in
\cref{sec:open}.  The appendices contain the certificate format and a few
routine proofs.

\section{Setup}\label{sec:setup}

\begin{tcolorbox}[colback=gray!4,colframe=gray!45,boxrule=0.5pt,arc=2pt,
  left=6pt,right=6pt,top=4pt,bottom=4pt,title={\small Conventions},
  fonttitle=\bfseries,coltitle=black,colbacktitle=gray!15]
\small
\begin{itemize}[leftmargin=1.2em,itemsep=2pt,topsep=2pt]
\item Graphs are finite and simple, and they are read as \emph{confusability
graphs}: $x\sim y$ means that inputs $x$ and $y$ may be confused.  We write
$x\simeq y$ for ``$x=y$ or $x\sim y$''.
\item Distinct non-adjacent vertices must be told apart, so they are
represented by \emph{orthogonal} objects.  This is Lov\'asz's original
orientation \cite{Lovasz1979}.  Li and Zuiddam express the same relations
through graph complements.
\item $G\sqcup H$ is the disjoint union and $G\boxt H$ the strong product,
where $(x,y)\simeq(x',y')$ if and only if $x\simeq x'$ and $y\simeq y'$.
$K_0$ is the null graph, with no vertices, $K_1$ is a single vertex,
$\overline{K_n}$ is the edgeless graph on $n$
vertices (a noiseless channel with $n$ inputs), and $C_5$ is the pentagon.
\item Logarithms are to base two unless written $\ln$.  Hilbert spaces are
finite-dimensional, of arbitrary dimension.  A state is \emph{faithful} if it
is invertible.  $\Pi_\rho$ is the support projection of $\rho$.
\end{itemize}
\end{tcolorbox}

\subsection{Representations and R\'enyi radii}

\begin{definition}\label[definition]{def:representation}
A \emph{representation} of $G$ is a family $(P_x)_{x\in V(G)}$ of nonzero
orthogonal projections on a Hilbert space such that $P_xP_y=0$ whenever
$x\ne y$ are non-adjacent.  A family of states $(\rho_x)$ is
\emph{compatible with} $G$ if its supports form a representation.
\end{definition}

A compatible family is the same thing as a classical--quantum channel
$x\mapsto\rho_x$ whose confusability graph is contained in $G$: two inputs can
only be confused if they are adjacent.  Extra orthogonality between adjacent
vertices is allowed.  Duan and Winter optimize over the same class of channels.
Readers who prefer channels whose confusability graph is exactly $G$ lose
nothing, because each of the infima
\eqref{eq:def-parameter}--\eqref{eq:sigma-def} below is the same for both
classes (\cref{lem:exact-graph}).

For a state $\rho$ and a faithful state $\sigma$, the Petz--R\'enyi
divergence of order $\alpha\in(0,1)\cup(1,2]$ is
\[
 D_\alpha(\rho\Vert\sigma)=\frac{1}{\alpha-1}\log\Tr\rho^\alpha\sigma^{1-\alpha}.
\]
At $\alpha=1$ it is the relative entropy $D(\rho\Vert\sigma)$, and at
$\alpha=0$ it is $D_0(\rho\Vert\sigma)=-\log\Tr\Pi_\rho\sigma$.  If $\sigma$
is a state that is not faithful but $\supp\rho\subseteq\supp\sigma$, the same
formulas are used with all powers and logarithms of $\sigma$ taken on its
support; \cref{lem:flagged} uses this convention.  We use four
standard facts \cite{Petz1986,Tomamichel2016}.  The divergence is
nonnegative and vanishes at $\rho=\sigma$.  It is additive under tensor
products.  It is nondecreasing in $\alpha$.  And it satisfies data processing
for every $\alpha\in[0,2]$.

\begin{definition}\label[definition]{def:theta-alpha}
For a graph $G$ and $\alpha\in[0,2]$, define $\vt_\alpha(G)$ by
\begin{equation}\label{eq:def-parameter}
 \log\vt_\alpha(G)
 =\inf\Bigl\{\max_{x}D_\alpha(\rho_x\Vert\sigma)\ :\
   (\rho_x)\text{ compatible with }G,\ \sigma\text{ faithful}\Bigr\},
\end{equation}
where the states live on a common Hilbert space of any finite dimension.  For
the null graph we set $\vt_\alpha(K_0)=0$.
\end{definition}

Nothing changes if singular centers are allowed with the usual support
conventions: mix a singular center with a little of the maximally mixed state
and let the weight go to zero.

\subsection{The Duan--Winter quantities}\label{sec:dw-quantities}

Duan and Winter attach the following quantities to a graph
\cite{DuanWinter2016}.  The \emph{minimal Holevo capacity} is
\begin{equation}\label{eq:cmin-def}
 \Cmin(G)=\inf_{(\rho_x)\text{ compatible}}\ \max_p
  \Bigl[S\Bigl(\sum_xp_x\rho_x\Bigr)-\sum_xp_xS(\rho_x)\Bigr].
\end{equation}
The \emph{no-signalling simulation cost} is
\begin{equation}\label{eq:sigma-def}
 \Sigma(G)=\inf\bigl\{\Tr T\ :\ T\ge\rho_x\text{ for all }x,\
   (\rho_x)\text{ compatible}\bigr\}.
\end{equation}
Duan and Winter show that $\min\{\Tr T:T\ge\rho_x\ \forall x\}$ is the real
SDP value governing exact simulation of the channel $x\mapsto\rho_x$ with
no-signalling assistance.  For a fixed channel, the literal one-shot
noiseless message alphabet has size equal to the ceiling of this SDP value;
the ceiling disappears after division by block length.  At graph level,
$\Sigma(G)$ is an infimum over compatible channels and is not known to be
attained.
Thus $\log\Sigma(G)$ is the infimum of the corresponding single-letter SDP
costs.  Tensoring witnesses shows that $\Sigma$ is
submultiplicative under $\boxt$, so the regularization
\[
 \SNS(G)=\lim_{n\to\infty}\frac1n\log\Sigma(G^{\boxt n})
\]
exists.  This is Duan and Winter's graph-level asymptotic simulation cost.
Finally, $\alpha^*(G)$ is Shannon's fractional packing number
\cite{Shannon1956}: the largest
$\sum_xv_x$ over $v\ge0$ with $\sum_{x\in C}v_x\le1$ for every clique $C$.  By
LP duality it is also the smallest $\sum_Cw_C$ over $w\ge0$ with
$\sum_{C\ni x}w_C\ge1$ for every vertex $x$.

\subsection{The entanglement-assisted spectrum}

We use the entanglement-assisted cohomomorphism preorder of Li and Zuiddam
\cite{LiZuiddam2020}, which comes from the entanglement-assisted
homomorphisms of Cubitt et al.~\cite{CMRSSW2014}.  In confusability language,
$G\le_*H$ means that there are a state $\omega$ and positive semidefinite
operators $\omega_g^h$ ($g\in V(G)$, $h\in V(H)$) with
\begin{align}
 \sum_{h\in V(H)}\omega_g^h&=\omega\quad\text{for every }g,\label{eq:ea-sum}\\
 \omega_g^h\,\omega_{g'}^{h'}&=0\quad\text{whenever $g\ne g'$ are non-adjacent
 in $G$ and $h\simeq h'$ in $H$.}\label{eq:ea-orth}
\end{align}
\Cref{lem:preorder} derives~\eqref{eq:ea-sum}--\eqref{eq:ea-orth} from
Li and Zuiddam's definition, which is phrased with graph complements.
Here is the operational picture.  Alice and Bob share a pure entangled state,
and $\omega$ is Bob's half of it.  To send $g$, Alice measures her half with a
measurement that depends on $g$, obtains an outcome $h$, and sends $h$ through
a channel with confusability graph $H$.  Her measurement steers Bob's half to
$\omega_g^h$ (unnormalized), and~\eqref{eq:ea-sum} says that Bob's marginal
does not depend on $g$.  Suppose $g$ and $g'$ must be told apart but the
channel might confuse $h$ with $h'$.  Then~\eqref{eq:ea-orth} lets Bob tell
them apart anyway, by measuring his own half.  So $G\le_*H$ says that, with
entanglement, one use of an $H$-channel does everything one use of a
$G$-channel does.  For example, $\overline{K_n}\le_*G$ says that $n$ messages
can be sent without error over one use of $G$.

\begin{definition}
The \emph{entanglement-assisted asymptotic spectrum} $\Xs$ is the set of maps
$\phi$ from finite graphs to $[0,\infty)$ such that
\[
 \phi(G\sqcup H)=\phi(G)+\phi(H),\quad
 \phi(G\boxt H)=\phi(G)\phi(H),\quad
 \phi(K_1)=1,\quad
 G\le_*H\Rightarrow\phi(G)\le\phi(H).
\]
\end{definition}

Li and Zuiddam show that $\Theta_*(G)=\min_{\phi\in\Xs}\phi(G)$, and that
$\vt\in\Xs$ \cite{LiZuiddam2020}.

\section{Power compressions and the shape of the family}
\label{sec:compressions}

The definition~\eqref{eq:def-parameter} optimizes over states.  It is much
easier to work with projections, and the tool for this is the following
compression.

\subsection{From states to projections}

Let $P\ne0$ be a projection, let $\sigma$ be positive definite (usually a
faithful state, but \cref{lem:tower} also compresses unnormalized operators),
and let $s\in[-1,1]$.  On $\ran P$ define
\[
 K_s(\sigma,P)=
 \begin{cases}(P\sigma^sP)^{1/s}, & s\ne0,\\ \exp\bigl(P(\ln\sigma)P\bigr), & s=0,\end{cases}
 \qquad Z_s(\sigma,P)=\Tr K_s(\sigma,P),
\]
extended by zero off $\ran P$.  We call $K_s$ the \emph{power compression} of
$\sigma$ to $P$ and $Z_s$ its \emph{weight}.  Four values of $s$ recur.
$Z_1=\Tr P\sigma$ is the plain weight.  $Z_{1/2}=\nrm{P\sigma^{1/2}P}_2^2$ is
a Hilbert--Schmidt weight.  $Z_0$ is a Gibbs weight.  And
$Z_{-1}=\Tr(P\sigma^{-1}P)^{-1}$ is the trace of the \emph{shorted operator}
$\sigma_{/P}$, the largest positive operator below $\sigma$ with range in
$\ran P$ \cite{AndersonTrapp1975}.  All the weights are homogeneous of degree
one: $K_s(c\sigma,P)=c\,K_s(\sigma,P)$ for $c>0$.

\begin{proposition}[supported minimization]\label[proposition]{prop:supported}
Let $0\le\alpha\le2$ and $s=1-\alpha$.  Then
\[
 \min_{\rho:\ \supp\rho\subseteq\ran P}D_\alpha(\rho\Vert\sigma)
 =-\log Z_s(\sigma,P),
\]
and the minimum is attained at $\rho=K_s(\sigma,P)/Z_s(\sigma,P)$.
\end{proposition}

\begin{proof}
For $\rho$ supported in $P$ we have
$\Tr\rho^\alpha\sigma^s=\Tr\rho^\alpha A$ with $A=P\sigma^sP$, and
$Z_s=\Tr A^{1/s}$ when $s\ne0$.

\emph{Case $0<\alpha<1$.}  Here $s\in(0,1)$.  H\"older's inequality for
Schatten norms gives
$\Tr\rho^\alpha A\le\nrm{\rho^\alpha}_{1/\alpha}\nrm{A}_{1/s}
=(\Tr A^{1/s})^s$.  Since $\alpha-1<0$, this becomes
$D_\alpha(\rho\Vert\sigma)\ge\frac{s}{\alpha-1}\log Z_s=-\log Z_s$, with
equality at $\rho\propto A^{1/s}$.

\emph{Case $1<\alpha\le2$.}  Now $s\in[-1,0)$ and $A$ is positive definite on
$\ran P$.  Write $A=\sum_ka_k\proj{e_k}$ and $q_k=\bra{e_k}\rho\ket{e_k}$, a
probability vector.  Convexity of $t\mapsto t^\alpha$ gives
$\bra{e_k}\rho^\alpha\ket{e_k}\ge q_k^\alpha$, so
$\Tr\rho^\alpha A\ge\sum_ka_kq_k^\alpha$.  By H\"older,
\[
 1=\sum_k\bigl(a_k^{1/\alpha}q_k\bigr)a_k^{-1/\alpha}
 \le\Bigl(\sum_ka_kq_k^\alpha\Bigr)^{1/\alpha}
    \Bigl(\sum_ka_k^{1/s}\Bigr)^{(\alpha-1)/\alpha},
\]
using $-1/(\alpha-1)=1/s$.  Hence
$\Tr\rho^\alpha A\ge(\sum_ka_k^{1/s})^{s}=Z_s^s$, and dividing the logarithm
by $\alpha-1>0$ gives $D_\alpha\ge-\log Z_s$.  Equality holds when $\rho$ is
diagonal in $(e_k)$ with $q_k\propto a_k^{1/s}$, that is, at
$\rho\propto A^{1/s}$.

\emph{Case $\alpha=1$.}  On $\ran P$ we have
$\ln K_0=P(\ln\sigma)P$, so $\Tr\rho\ln\sigma=\Tr\rho\ln(K_0/Z_0)+\ln Z_0$ and
$D(\rho\Vert\sigma)=D(\rho\Vert K_0/Z_0)-\log Z_0\ge-\log Z_0$.

\emph{Case $\alpha=0$.}  $\Tr\Pi_\rho\sigma\le\Tr P\sigma=Z_1$, with equality
whenever $\rho$ has full support in $\ran P$.
\end{proof}

The minimizers for different vertices have the supports we prescribe, so they
form a compatible family whenever the projections form a representation.
Conversely, the supports of any compatible family form a representation.  This
proves the projection formula:
\begin{equation}\label{eq:projection-formula}
 \frac1{\vt_\alpha(G)}
 =\sup_{(P_x),\,\sigma}\ \min_{x\in V(G)}Z_{1-\alpha}(\sigma,P_x),
\end{equation}
with the supremum over representations of $G$ and faithful states $\sigma$ on
the same space.

\subsection{Endpoints, flat witnesses and a sandwich}

Some witnesses have the same radius at every order.  Call a pair
$(\rho,\sigma)$ \emph{flat} if $\sigma$ commutes with a projection $\Pi$ and
$\rho=\Pi\sigma/\Tr\Pi\sigma$.

\begin{lemma}[flat witnesses]\label[lemma]{lem:flat}
If $(\rho,\sigma)$ is flat, then $D_\alpha(\rho\Vert\sigma)=-\log\Tr\Pi\sigma$
for every $\alpha\in[0,2]$.
\end{lemma}

\begin{proof}
With $c=\Tr\Pi\sigma$ we get
$\Tr\rho^\alpha\sigma^{1-\alpha}=c^{-\alpha}\Tr\Pi\sigma=c^{1-\alpha}$, which
gives the claim for $\alpha\ne0,1$.  The two remaining orders are direct.
\end{proof}

\begin{proposition}\label[proposition]{prop:endpoints}
For every graph $G$:
\begin{enumerate}[label=\textup{(\roman*)},itemsep=1pt]
\item $\vt_0(G)=\vt(G)$ \textup{\cite[Theorem~8]{Dalai2013TIT},
\cite[proof of Theorem~5]{DuanWinter2016}},
$\vt_1(G)=2^{\Cmin(G)}$
\textup{\cite{OhyaPetzWatanabe1997,SchumacherWestmoreland2001}} and
$\vt_2(G)=\Sigma(G)$;
\item $\alpha\mapsto\vt_\alpha(G)$ is nondecreasing on $[0,2]$
\textup{(immediate from the monotonicity of $D_\alpha$ in $\alpha$
\cite{Tomamichel2016})};
\item $\vt(G)\le\vt_\alpha(G)\le\alpha^*(G)$ for every $\alpha\in[0,2]$; at
$\alpha\in\{0,1,2\}$ these bounds follow from Duan and Winter's
chain~\eqref{eq:dw-chain} \textup{\cite{DuanWinter2016}}.
\end{enumerate}
\end{proposition}

\begin{proof}
(i) \emph{Order zero.}  This is Dalai's Theorem~8, which he credits to
Schrijver \cite{Dalai2013TIT}.  Duan and Winter give the same purification
argument in the proof of their Theorem~5 \cite[Section~V]{DuanWinter2016},
and for a pure center the passage from projections to vectors is the
argument of Cabello, Severini and Winter \cite[proof of Eq.~(6)]{CSW2010}.
We include the short argument in our notation.
By~\eqref{eq:projection-formula},
$1/\vt_0(G)=\sup\min_x\Tr P_x\sigma$.  Lov\'asz's definition of $\vt$
through the value of an orthonormal representation reads
$1/\vt(G)=\sup\min_x|\langle c,u_x\rangle|^2$ over orthonormal
representations $(u_x)$ of $G$ and unit vectors $c$
\cite[Section~II]{Lovasz1979}.  Lov\'asz uses real vectors, but complex ones
give the same value: rephase each $u_x$ so that $\langle c,u_x\rangle\ge0$,
and then realify, $v\mapsto(\operatorname{Re}v,\operatorname{Im}v)$, which
preserves unit length, orthogonality and the values $\langle c,u_x\rangle$.  A
vector representation is a projection representation with $P_x=\proj{u_x}$,
and a faithful approximation of $\proj c$ gives $\vt_0\le\vt$.  Conversely,
purify $\sigma$ to a unit vector $c\in\cH\otimes\cH'$ and put
$u_x\propto(P_x\otimes I)c$.  Non-adjacent vertices get orthogonal vectors,
and $|\langle c,u_x\rangle|^2=\Tr P_x\sigma$, so $\vt\le\vt_0$.

\emph{Order one.}  For each compatible family, the Holevo capacity equals the
relative-entropy radius $\min_\sigma\max_xD(\rho_x\Vert\sigma)$
\cite{OhyaPetzWatanabe1997,SchumacherWestmoreland2001}.  Taking the infimum
over families gives $\vt_1=2^{\Cmin}$.

\emph{Order two.}  By~\eqref{eq:projection-formula},
$1/\vt_2=\sup\min_x\Tr\sigma_{/P_x}$.  Given a compatible family and $T\ge\rho_x$
with $t=\Tr T$, put $\sigma=T/t$, regularized if needed.  Each $\rho_x/t$ lies
below $\sigma$ with range in $\ran P_x$, so $\Tr\sigma_{/P_x}\ge1/t$.  Hence
$\vt_2\le\Sigma$.  Conversely, if $\Tr\sigma_{/P_x}\ge c$ for all $x$, the
states $\rho_x=\sigma_{/P_x}/\Tr\sigma_{/P_x}$ are compatible and all lie below
$T=\sigma/c$, whose trace is $1/c$.  Hence $\Sigma\le\vt_2$.

(ii) $D_\alpha$ is nondecreasing in $\alpha$ for each pair of states, and
infima preserve this.

(iii) The lower bound follows from (i) and (ii).  The upper bound also
follows from (i), (ii) and~\eqref{eq:dw-chain}, since
$\vt_\alpha\le\vt_2=\Sigma\le\alpha^*$; flat witnesses give it directly.
Take a
fractional clique cover $w$ with $\sum_Cw_C=\alpha^*(G)$, and work on the
space spanned by orthonormal vectors $\ket C$, one for each clique with
$w_C>0$.  Put $\sigma=\sum_Cw_C\proj C/\alpha^*(G)$ and let $\Pi_x$ project
onto the cliques containing $x$.  Non-adjacent vertices share no clique, so the
flat states $\rho_x=\Pi_x\sigma/\Tr\Pi_x\sigma$ are compatible.  By
\cref{lem:flat} their radius at every order is
$\max_x\log\bigl(\alpha^*(G)/\sum_{C\ni x}w_C\bigr)\le\log\alpha^*(G)$.
\end{proof}

Two consequences are worth stating.  First, if $\vt(G)=\alpha^*(G)$, the
whole family collapses on $G$.  This happens for every perfect graph.  Second,
if some flat witness is optimal for $\vt_1(G)$, then $\vt_1(G)=\vt_2(G)$,
because the flat radius does not depend on the order.  On the pentagon, this
is exactly what happens (\cref{sec:pentagon}).

\subsection{Continuity}

\begin{proposition}\label[proposition]{prop:continuity}
For every graph $G$, the map $\alpha\mapsto\vt_\alpha(G)$ is continuous on
$[0,1)$.
\end{proposition}

Right continuity comes from fixing a near-optimal witness.  Left continuity
comes from a comparison $D_\alpha\le\frac{\alpha(1-\beta)}{\beta(1-\alpha)}D_\beta$
for $\beta<\alpha<1$, whose constant does not depend on the states or the
dimension.  The details are in \cref{app:continuity}.  The right-continuity
argument works at every order in $[0,2)$, so the map is right-continuous on
$[0,2)$.  Left continuity at orders in $[1,2]$ is not proved here, except on
the pentagon, where the map is constant on $[\frac12,2]$
(\cref{thm:pentagon-half}).

\subsection{Sandwiched radii give the same family}\label{sec:sandwiched}

The Petz divergence is not the only quantum R\'enyi divergence.  For
$\beta\in[\frac12,1)\cup(1,\infty)$ the \emph{sandwiched} R\'enyi divergence
\cite{MDSFT2013,WWY2014} is
\[
 \widetilde D_\beta(\rho\Vert\sigma)
 =\frac1{\beta-1}\log\Tr\Bigl(\sigma^{\frac{1-\beta}{2\beta}}\,\rho\,
   \sigma^{\frac{1-\beta}{2\beta}}\Bigr)^{\beta},
\]
with the relative entropy $\widetilde D_1=D$ at order one and the
max-relative entropy
$\widetilde D_\infty(\rho\Vert\sigma)=\log\min\{\lambda:\rho\le\lambda\sigma\}$
at order infinity.  This family satisfies data processing throughout
$[\frac12,\infty]$ \cite{FrankLieb2013,Beigi2013}.  Define
$\widetilde\vt_\beta(G)$
by~\eqref{eq:def-parameter} with $\widetilde D_\beta$ in place of $D_\alpha$,
using the same compatible families and faithful centers, and set
$\widetilde\vt_\beta(K_0)=0$.  On this full range it gives exactly the Petz
family, with a change of order.

\begin{proposition}[sandwiched radii]\label[proposition]{prop:sandwiched}
For every graph $G$ and every $\beta\in[\frac12,\infty]$,
\[
 \widetilde\vt_\beta(G)=\vt_{2-1/\beta}(G),
\]
with $2-1/\infty=2$.  In particular,
\[
 \widetilde\vt_{1/2}=\vt,\qquad
 \widetilde\vt_{2/3}=\vt_{1/2},\qquad
 \widetilde\vt_1=2^{\Cmin},\qquad
 \widetilde\vt_\infty=\Sigma.
\]
\end{proposition}

\begin{proof}
Put $s=(1-\beta)/\beta\in[-1,1]$, with $s=-1$ at $\beta=\infty$, so the
corresponding Petz order is $\alpha=1-s=2-1/\beta$.  As in the proof
of~\eqref{eq:projection-formula}, it suffices to show that for a projection
$P\ne0$ and a faithful state $\sigma$,
\begin{equation}\label{eq:sandwiched-supported}
 \min_{\rho:\ \supp\rho\subseteq\ran P}\widetilde D_\beta(\rho\Vert\sigma)
 =-\log Z_s(\sigma,P),
\end{equation}
with a minimizer faithful on $\ran P$.  All inverses and compressed powers
below are taken on $\ran P$.

\emph{Case $\frac12\le\beta<1$.}  Here $s\in(0,1]$.  Put
$A=P\sigma^sP>0$ on $\ran P$ and $X=A^{1/2}\rho A^{1/2}$.  The constraint
on $\rho$ becomes $X\ge0$, $\Tr A^{-1}X=1$.  The map
$V=\sigma^{s/2}P A^{-1/2}:\ran P\to\cH$ is an isometry, and
\[
 \sigma^{s/2}\rho\sigma^{s/2}=VXV^*.
\]
Consequently the sandwiched moment is $\Tr X^\beta$.
Pinch $X$ in an orthonormal eigenbasis $(e_i)$ of $A$, obtaining $Y$.
Pinching preserves $\Tr A^{-1}X$ and cannot decrease $\Tr X^\beta$:
it is an average of unitary conjugations, and $X\mapsto\Tr X^\beta$ is
concave for $0<\beta<1$.  Writing $Ae_i=a_i e_i$ and
$\langle e_i,Ye_i\rangle=a_iq_i$, we have $q_i\ge0$ and $\sum_iq_i=1$.
Scalar H\"older with conjugate exponents $1/\beta$ and $1/(1-\beta)$ gives
\[
 \Tr X^\beta\le\Tr Y^\beta
 =\sum_i a_i^\beta q_i^\beta
 \le\Bigl(\sum_i a_i^{\beta/(1-\beta)}\Bigr)^{1-\beta}
 =Z_s(\sigma,P)^{1-\beta}.
\]
Equality holds at $\rho=A^{1/s}/Z_s=K_s/Z_s$:
then $q_i=a_i^{1/s}/Z_s$ and $X$ is already diagonal in $(e_i)$.
Since $\beta-1<0$, maximizing the moment minimizes the divergence, giving
\eqref{eq:sandwiched-supported}.  At $\beta=\frac12$ this directly reads
$-\log Z_1(\sigma,P)=-\log\Tr P\sigma$; no limiting argument is needed.

\emph{Case $\beta=1$.}  This is Proposition~\ref{prop:supported} at order
one, with minimizer $K_0/Z_0$.

\emph{Case $1<\beta<\infty$.}  Put $\theta=(\beta-1)/\beta\in(0,1)$,
so $s=-\theta$, and let $A=P\sigma^{-\theta}P$, which is positive
definite on $\ran P$, and $X=A^{1/2}\rho A^{1/2}$.  The operator
$\sigma^{-\theta/2}\rho\,\sigma^{-\theta/2}$ has the same nonzero eigenvalues
as $\rho^{1/2}\sigma^{-\theta}\rho^{1/2}=\rho^{1/2}A\rho^{1/2}$, and hence as
$X$, so the sandwiched moment is $\Tr X^\beta$.  Since
$\Tr A^{-1}X=\Tr\rho=1$, H\"older's inequality with the conjugate exponents
$\beta$ and $\beta/(\beta-1)=1/\theta$ gives
\[
 1\le\nrm X_\beta\,\nrm{A^{-1}}_{1/\theta}
 =\bigl(\Tr X^\beta\bigr)^{1/\beta}\bigl(\Tr A^{-1/\theta}\bigr)^{\theta}.
\]
Here $A^{-1/\theta}=(P\sigma^{-\theta}P)^{1/(-\theta)}=K_{-\theta}(\sigma,P)$,
so $\Tr X^\beta\ge Z_{-\theta}(\sigma,P)^{-(\beta-1)}$, which is the lower bound
in~\eqref{eq:sandwiched-supported}.  At $\rho=K_{-\theta}/Z_{-\theta}$ we get
$X\propto A^{1-1/\theta}$, so $X^\beta\propto A^{-1/\theta}=(A^{-1})^{1/\theta}$,
which is the equality case of H\"older's inequality.

\emph{Case $\beta=\infty$.}  The condition $\rho\le\lambda\sigma$ for a state
supported in $\ran P$ says that $\rho/\lambda$ lies below $\sigma$ with
range in $\ran P$, so
$1/\lambda\le\Tr\sigma_{/P}=Z_{-1}(\sigma,P)$, with equality at
$\rho=\sigma_{/P}/\Tr\sigma_{/P}=K_{-1}/Z_{-1}$.

In every case $K_s/Z_s$ is faithful on $\ran P$.  Thus the minimizing
states retain the prescribed supports, and for a graph with at least one
vertex the same
argument as in~\eqref{eq:projection-formula} yields
\[
 \frac1{\widetilde\vt_\beta(G)}
 =\sup_{(P_x),\,\sigma}\min_x Z_s(\sigma,P_x)
 =\frac1{\vt_{1-s}(G)}.
\]
The null graph follows from the conventions.  These identities hold in
each fixed dimension before optimizing over representations and centers;
no exchange of an order limit with an unbounded-dimension infimum is used.
\end{proof}

Thus the sandwiched radii on $[\frac12,\infty]$ give exactly the Petz family
on $[0,2]$, from the Lov\'asz number to the no-signalling simulation cost.
By \cref{thm:A}, every
$\widetilde\vt_\beta$ in this range is a point of $\Xs$, and the two
families give the same set of spectral points.  This is an identity of
optimized graph radii after reparameterization, not an identity between
the two divergences for arbitrary fixed states.  For a pure state
$\psi=\proj\phi$ and a faithful state $\sigma$ the two divergences do agree:
$\widetilde D_\beta(\psi\Vert\sigma)=D_{2-1/\beta}(\psi\Vert\sigma)$ for
every $\beta\in[\frac12,\infty]$.  Indeed, for $\beta\ne1,\infty$ the
operator $\sigma^{(1-\beta)/2\beta}\psi\,\sigma^{(1-\beta)/2\beta}$ has rank
one, so both sides equal
$\frac{\beta}{\beta-1}\log\bra\phi\sigma^{(1-\beta)/\beta}\ket\phi$; at
$\beta=1$ both are the relative entropy, and at $\beta=\infty$ both equal
$\log\bra\phi\sigma^{-1}\ket\phi$.  \Cref{prop:sandwiched} is the
corresponding statement for optimized radii, where the optimal states need
not be pure.  At $\beta=\infty$ the
identity $\widetilde\vt_\infty=\Sigma$ is Duan and Winter's conditional
min-entropy formula \cite[Theorem~2]{DuanWinter2016}, written as a
max-relative-entropy radius and optimized over compatible channels.

\subsection{The tower identity}

The key algebraic fact is that compressing in two steps is the same as
compressing in one.

\begin{lemma}[tower identity]\label[lemma]{lem:tower}
Let $P\le Q$ be projections and $s\in[-1,1]$.  Compute the outer compression
below inside $\ran Q$, where $K_s(\sigma,Q)$ is faithful.  Then
\[
 K_s\bigl(K_s(\sigma,Q),P\bigr)=K_s(\sigma,P).
\]
Consequently, if $z=Z_s(\sigma,Q)$, then
$Z_s\bigl(K_s(\sigma,Q)/z,\,P\bigr)=Z_s(\sigma,P)/z$.
\end{lemma}

\begin{proof}
For $s\ne0$, the $s$-th power of $K_s(\sigma,Q)$ inside $\ran Q$ is
$Q\sigma^sQ$, and compressing further to $P\le Q$ gives $P\sigma^sP$.  For
$s=0$ the same computation works with logarithms.  The second statement
follows from homogeneity.
\end{proof}

\section{Every order up to two is a spectral point}\label{sec:spectral}

We need four properties: normalization, additivity under $\sqcup$,
multiplicativity under $\boxt$, and monotonicity under $\le_*$.
Normalization is immediate.  On $K_1$ the divergence is nonnegative and
vanishes at $\rho=\sigma$, so $\vt_\alpha(K_1)=1$.  The product is the heart
of the matter, so we start there.

\subsection{Products}

\begin{figure}[t]
\centering
\begin{tikzpicture}[
  cell/.style={draw=blue!50!black, fill=blue!7, minimum width=1.05cm,
               minimum height=0.62cm, font=\scriptsize, inner sep=0pt},
  >=Latex]
\foreach \i in {1,2,3} {
  \foreach \j in {1,2,3,4} {
    \node[cell] (c\i\j) at (1.25*\j, -0.95*\i) {$P_{x_{\i}y_{\j}}$};
  }
  \node[font=\small, anchor=east] at (0.45,-0.95*\i) {$x_{\i}$};
  \draw[rounded corners=4pt, dashed, thick, teal!60!black]
     (0.58,-0.95*\i+0.42) rectangle (5.6,-0.95*\i-0.42);
  \node[font=\small, teal!50!black, anchor=west] at (5.65,-0.95*\i) {$S_{x_{\i}}$};
}
\foreach \j in {1,2,3,4} \node[font=\small] at (1.25*\j,-0.25) {$y_{\j}$};
\node[font=\scriptsize, align=left, anchor=north west, text width=6.8cm]
  at (0.1,-3.45)
  {Distinguishable rows are orthogonal: $(S_x)$ represents $G$.\\
   Within one row, $(P_{xy})_y$ represents $H$.};
\begin{scope}[xshift=8.3cm, yshift=-0.6cm]
  \node (s) at (0,0) {$\sigma$};
  \node (m) at (3.7,0) {$K_s(\sigma,S_x)$};
  \node (p) at (3.7,-2.3) {$K_s(\sigma,P_{xy})$};
  \draw[->, thick] (s) -- node[above, font=\scriptsize] {compress to $S_x$} (m);
  \draw[->, thick] (m) -- node[right, font=\scriptsize, align=left]
       {compress to\\ $P_{xy}\le S_x$} (p);
  \draw[->, thick] (s) -- node[below left, font=\scriptsize, align=right]
       {compress to $P_{xy}$\\ directly} (p);
  \node[font=\scriptsize, align=center, text width=5.4cm] at (2.0,-3.4)
       {Both routes agree (\cref{lem:tower}).  With
        $\sigma_x=K_s(\sigma,S_x)/z_x$ this reads
        $Z_s(\sigma_x,P_{xy})=Z_s(\sigma,P_{xy})/z_x$.};
\end{scope}
\end{tikzpicture}
\caption{The proof of \cref{thm:product}.  Left: a representation of
$G\boxt H$ that need not be a tensor product, organized by rows.  Right: the
tower identity lets each row be analysed with its own reduced center, so the
weight of a cell factors into a $G$-part $z_x$ and an $H$-part.}
\label{fig:product}
\end{figure}
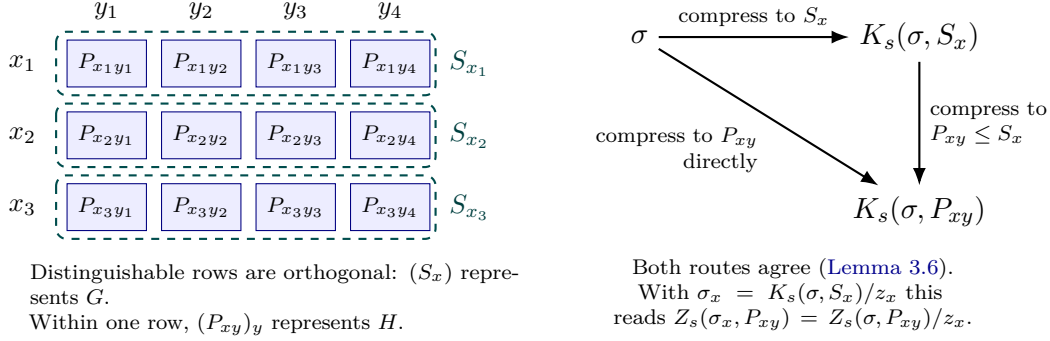

Here is the idea (\cref{fig:product}).  Take any representation $(P_{xy})$ of
$G\boxt H$ and group its projections by the first coordinate.  The ranges in
row $x$ span a subspace, with projection $S_x$.  If $x$ and $x'$ are
distinguishable in $G$, then every cell of row $x$ is distinguishable from
every cell of row $x'$, so $S_x\perp S_{x'}$ and the rows represent $G$.
Inside one row, the cells represent $H$.  All that is missing is a way to
split the weight of a cell into a row part and a within-row part.  The tower
identity does exactly this.

\begin{theorem}[products]\label{thm:product}
For $0\le\alpha\le2$ and all graphs $G,H$,
\[
 \vt_\alpha(G\boxt H)=\vt_\alpha(G)\,\vt_\alpha(H).
\]
\end{theorem}

\begin{proof}
If either graph is the null graph both sides vanish, so assume that both
have vertices and put $s=1-\alpha$.

For ``$\le$'', take witnesses $(\rho_x),\sigma$ for $G$ and $(\tau_y),\omega$
for $H$.  If $(x,y)$ and $(x',y')$ are distinct and non-adjacent in $G\boxt H$,
then $x,x'$ are distinct and non-adjacent in $G$ or $y,y'$ are in $H$.  In
either case $(\rho_x\otimes\tau_y)(\rho_{x'}\otimes\tau_{y'})=0$.  Since
$D_\alpha$ is additive under tensor products, radii add.

For ``$\ge$'', let $(P_{xy})$ represent $G\boxt H$, let $\sigma$ be faithful,
and suppose $Z_s(\sigma,P_{xy})\ge c>0$ for every pair $(x,y)$.  Let $S_x$ be
the projection onto $\sum_y\ran P_{xy}$, and put $z_x=Z_s(\sigma,S_x)$.

\emph{The rows represent $G$.}  If $x\ne x'$ are non-adjacent in $G$, every
pair $(x,y),(x',y')$ is distinct and non-adjacent in $G\boxt H$, so
$P_{xy}P_{x'y'}=0$ and hence $S_xS_{x'}=0$.

\emph{Each row represents $H$.}  Fix $x$.  If $y\ne y'$ are non-adjacent in
$H$, then $(x,y)$ and $(x,y')$ are non-adjacent, so $(P_{xy})_y$ represents
$H$ on $\ran S_x$.  Use the faithful state $\sigma_x=K_s(\sigma,S_x)/z_x$ on
$\ran S_x$ as center.  By \cref{lem:tower},
$Z_s(\sigma_x,P_{xy})=Z_s(\sigma,P_{xy})/z_x\ge c/z_x$ for every $y$.  The
projection formula~\eqref{eq:projection-formula} for $H$ gives
$c/z_x\le1/\vt_\alpha(H)$, that is, $z_x\ge c\,\vt_\alpha(H)$.

\emph{Combine.}  Apply~\eqref{eq:projection-formula} for $G$ to $(S_x)$ and
$\sigma$: $c\,\vt_\alpha(H)\le\min_xz_x\le1/\vt_\alpha(G)$.  So
$c\le1/(\vt_\alpha(G)\vt_\alpha(H))$ for every admissible $c$, and taking the
supremum gives
\[
 \frac1{\vt_\alpha(G\boxt H)}\le\frac1{\vt_\alpha(G)\,\vt_\alpha(H)} .
 \qedhere
\]
\end{proof}

\begin{remark}\label[remark]{rem:single-letter}
At $\alpha=0$ and $\alpha=1$, \cref{thm:product} recovers the multiplicativity
of $\vt$ \cite{Lovasz1979} and Duan and Winter's additivity of $\Cmin$
\cite[Lemma~23]{DuanWinter2016}.  At $\alpha=2$ it says
$\Sigma(G\boxt H)=\Sigma(G)\Sigma(H)$.
Hence $\SNS(G)=\log\Sigma(G)$: regularizing the graph-level value
$\log\Sigma(G)$, an infimum of one-shot SDP costs, does not change it.  Their
fixed-cq-graph multiplicativity \cite[Proposition~18]{DuanWinter2016} gives
this identity only for the infimum restricted to product representations
$P_{xy}=P_x\otimes Q_y$.  \Cref{thm:product} removes that restriction, but it
does not recover their statement for fixed cq-graphs.
\end{remark}

\subsection{Disjoint unions}

Two small lemmas do the work.  The first rewrites a divergence after
compressing the center.  The second says that the weights of orthogonal blocks
never add up to more than one.

\begin{lemma}[compression identity]\label[lemma]{lem:compression}
Let $Q\ne0$ be a projection, $s=1-\alpha\in[-1,1]$, and
$\sigma_Q=K_s(\sigma,Q)/Z_s(\sigma,Q)$.  Every state $\rho$ supported in
$\ran Q$ satisfies
\[
 D_\alpha(\rho\Vert\sigma)=D_\alpha(\rho\Vert\sigma_Q)-\log Z_s(\sigma,Q).
\]
\end{lemma}

\begin{proof}
For $\alpha\notin\{0,1\}$, substitute $Q\sigma^sQ=Z_s^s\sigma_Q^s$ into the
Petz moment.  For $\alpha=1$ use $Q(\ln\sigma)Q=\ln\sigma_Q+(\ln Z_0)Q$ on
$\ran Q$, and for $\alpha=0$ use $Q\sigma Q=Z_1\sigma_Q$.
\end{proof}

\begin{lemma}[orthogonal blocks]\label[lemma]{lem:orthogonal-sum}
If $Q_1,\dots,Q_m$ are mutually orthogonal nonzero projections, $\sigma$ is a faithful
state and $s\in[-1,1]$, then $\sum_jZ_s(\sigma,Q_j)\le1$.
\end{lemma}

The proof is a one-line application of Jensen's inequality
(\cref{app:orthogonal-sum}).

\begin{theorem}[disjoint unions]\label{thm:union}
For $0\le\alpha\le2$, $\vt_\alpha(G\sqcup H)=\vt_\alpha(G)+\vt_\alpha(H)$.
\end{theorem}

\begin{proof}
Assume that neither graph is the null graph.  For ``$\le$'', put witnesses
for $G$ and $H$
with radii at most $\log a$ and $\log b$ on orthogonal summands, and use the
center $(a\sigma_G\oplus b\sigma_H)/(a+b)$.  Compressing this center to the
first summand gives weight $a/(a+b)$ and reduced center $\sigma_G$, so by
\cref{lem:compression} every $G$-state has radius at most
$\log a+\log\frac{a+b}{a}=\log(a+b)$.  The same holds for $H$.

For ``$\ge$'', take a compatible family of $G\sqcup H$ with center $\sigma$
and radius at most $R$.  Let $Q_G$ and $Q_H$ project onto the spans of the
supports of the two subfamilies.  Cross pairs are non-adjacent, so
$Q_GQ_H=0$.  By \cref{lem:compression}, the $G$-states have radius at most
$R+\log Z_s(\sigma,Q_G)$ around $\sigma_{Q_G}$, so
$\vt_\alpha(G)\le2^RZ_s(\sigma,Q_G)$.  Likewise
$\vt_\alpha(H)\le2^RZ_s(\sigma,Q_H)$.  Add the two bounds and apply
\cref{lem:orthogonal-sum}.
\end{proof}

In particular $\vt_\alpha(\overline{K_n})=n$.

\subsection{Entanglement-assisted monotonicity}

The proof follows the protocol described after~\eqref{eq:ea-orth}.  Bob ends
up holding his share of the entanglement together with the channel output.
If $\tau_h$ is the output for $h$, his joint state given $g$ is
$\sum_h\omega_g^h\otimes\tau_h$.  These joint states form a new channel for
$G$, and the only question is whether its radius can be larger than the old
one.  It cannot.  The joint state is a flagged mixture with the flag
discarded, and discarding the flag is data processing.

\begin{lemma}[flagged mixtures]\label[lemma]{lem:flagged}
Let $\omega_h\ge0$ with $\sum_h\omega_h=\omega$ and $\Tr\omega=1$, let
$\tau_h$ be states and let $\sigma$ be faithful.  If
$D_\alpha(\tau_h\Vert\sigma)\le R$ whenever $\omega_h\ne0$, then for
$0\le\alpha\le2$
\[
 D_\alpha\Bigl(\sum_h\omega_h\otimes\tau_h\ \Big\Vert\ \omega\otimes\sigma\Bigr)\le R .
\]
\end{lemma}

\begin{proof}
Add a classical flag and compare
$\Omega=\sum_h\proj h\otimes\omega_h\otimes\tau_h$ with
$\Omega'=\sum_h\proj h\otimes\omega_h\otimes\sigma$; the support of $\Omega$
lies in that of $\Omega'$.  The Petz moment splits over the flag, and the
$\omega_h$ factors cancel:
\[
 \Tr\Omega^\alpha\Omega'^{1-\alpha}=\sum_hw_h\Tr\tau_h^\alpha\sigma^{1-\alpha},
 \qquad w_h=\Tr\omega_h,\ \ \textstyle\sum_hw_h=1.
\]
For $0<\alpha<1$ each term satisfies
$\Tr\tau_h^\alpha\sigma^{1-\alpha}\ge2^{(\alpha-1)R}$, so the average does too
and $D_\alpha(\Omega\Vert\Omega')\le R$.  For $1<\alpha\le2$ all inequalities
reverse and the conclusion is the same.  At $\alpha=1$ the divergence is the
$w$-average of the $D(\tau_h\Vert\sigma)$, and at $\alpha=0$ we get
$\Tr\Pi_\Omega\Omega'=\sum_hw_h\Tr\Pi_{\tau_h}\sigma\ge2^{-R}$.  Tracing out
the flag maps $\Omega\mapsto\sum_h\omega_h\otimes\tau_h$ and
$\Omega'\mapsto\omega\otimes\sigma$, and Petz divergences of order
$\alpha\in[0,2]$ do not increase under partial traces
\cite{Petz1986,Tomamichel2016}.
\end{proof}

\begin{theorem}[entanglement-assisted monotonicity]\label{thm:ea-mono}
If $G\le_*H$, then $\vt_\alpha(G)\le\vt_\alpha(H)$ for every $\alpha\in[0,2]$.
\end{theorem}

\begin{proof}
If $G$ is the null graph there is nothing to prove, and otherwise
\eqref{eq:ea-sum} forces $H$ to have a vertex.  Take $\omega,\omega_g^h$ as
in~\eqref{eq:ea-sum}--\eqref{eq:ea-orth}, restricted to $\supp\omega$, and a
compatible family $(\tau_h)$ for $H$ with faithful center $\sigma$ and radius
at most $R$.  Define
\[
 \tau'_g=\sum_h\omega_g^h\otimes\tau_h,\qquad \sigma'=\omega\otimes\sigma .
\]
Each $\tau'_g$ has trace $\Tr\omega=1$.  If $g\ne g'$ are non-adjacent, every
term of $\tau'_g\tau'_{g'}=\sum_{h,h'}\omega_g^h\omega_{g'}^{h'}\otimes\tau_h\tau_{h'}$
vanishes.  When $h\simeq h'$ this is~\eqref{eq:ea-orth}.  Otherwise $h,h'$
are distinct and non-adjacent in $H$, and then $\tau_h\tau_{h'}=0$.  So
$(\tau'_g)$ is compatible with $G$, and \cref{lem:flagged} bounds its radius
around $\sigma'$ by $R$.
\end{proof}

\begin{proof}[Proof of \cref{thm:A}]
Normalization was noted at the start of the section.  Additivity is
\cref{thm:union}, multiplicativity is \cref{thm:product}, and monotonicity is
\cref{thm:ea-mono}.  At $\alpha=0$ these are properties of $\vt$, and
membership of $\vt$ in $\Xs$ is Li and Zuiddam's theorem
\cite{LiZuiddam2020}, whose monotonicity step is due to Cubitt et
al.~\cite{CMRSSW2014}.
\end{proof}

\begin{remark}[orders above two]\label[remark]{rem:above-two}
The only place the restriction $\alpha\le2$ enters is data processing in
\cref{lem:flagged}.  \Cref{prop:supported}, \cref{lem:tower} and
\cref{lem:compression,lem:orthogonal-sum} hold verbatim for $s<-1$, that is,
for $\alpha>2$, so data processing is the only obstruction in this argument.
Petz data processing fails above order two, and the inequality of
\cref{lem:flagged} itself fails too: the accompanying verification scripts
(see Data and code availability) exhibit an explicit finite-dimensional
instance at order three in which it fails by more than half a bit.  So this
argument says nothing about $\alpha>2$.
\end{remark}

\section{The pentagon}\label{sec:pentagon}

Throughout this section $C_5$ is labelled by $\mathbb Z_5$ with $i\sim i\pm1$,
and a representation satisfies $P_iP_{i+2}=0$ for every $i$.  We begin at
order one half, where everything is exact, and then work down to order zero.

\subsection{Order one half and above}

\begin{lemma}[Hilbert--Schmidt inequality]\label[lemma]{lem:hs}
For every representation of $C_5$ and every operator $X$,
\[
 \sum_i\nrm{P_iXP_i}_2^2\le2\nrm{X}_2^2 .
\]
\end{lemma}

\begin{proof}
Put $T=\sum_iP_i$.  The terms $P_iP_{i\pm2}P_i$ vanish, and $P_{i-1}$ and
$P_{i+1}$ are orthogonal to each other, so
\begin{equation}\label{eq:PTP}
 P_iTP_i=P_i+P_i(P_{i-1}+P_{i+1})P_i\le2P_i .
\end{equation}
Consider the map $\mathcal A:(X_i)_i\mapsto\sum_iX_i$ from
$\bigoplus_iP_i\mathcal B_2P_i$ to the Hilbert--Schmidt class $\mathcal B_2$.
For $X_i=P_iX_iP_i$, cyclicity gives
$\Tr X_i^*X_j=\Tr\bigl((P_jX_i)^*(X_jP_i)\bigr)$, so
\[
 |\Tr X_i^*X_j|\le\tfrac12\bigl(\nrm{P_jX_i}_2^2+\nrm{X_jP_i}_2^2\bigr).
\]
Summing over $i$ and $j$ and using~\eqref{eq:PTP},
\[
 \Bigl\lVert\sum_iX_i\Bigr\rVert_2^2
 \le\tfrac12\sum_i\bigl[\Tr(X_i^*TX_i)+\Tr(X_iTX_i^*)\bigr]
 \le2\sum_i\nrm{X_i}_2^2 .
\]
So $\nrm{\mathcal A}^2\le2$.  The adjoint of $\mathcal A$ is
$X\mapsto(P_iXP_i)_i$, and it has the same norm.
\end{proof}

\begin{theorem}\label{thm:pentagon-half}
$\vt_\alpha(C_5)=\frac52$ for every $\alpha\in[\frac12,2]$.
\end{theorem}

\begin{proof}
At $s=\frac12$ the weight is $Z_{1/2}(\sigma,P)=\nrm{P\sigma^{1/2}P}_2^2$.
\Cref{lem:hs} with $X=\sigma^{1/2}$ gives
$\sum_iZ_{1/2}(\sigma,P_i)\le2\Tr\sigma=2$.  So the minimum over $i$ is at
most $\frac25$, and~\eqref{eq:projection-formula} gives
$\vt_{1/2}(C_5)\ge\frac52$.  By monotonicity in $\alpha$ the same lower bound
holds on $[\frac12,2]$.  For the upper bound, $\alpha^*(C_5)=\frac52$ and
\cref{prop:endpoints}(iii) applies.  Explicitly, the flat witness is
$\rho_i=\frac12(\proj i+\proj{i+1})$ on $\mathbb C^5$ with $\sigma=I/5$, and
it has radius $\log\frac52$ at every order.
\end{proof}

\begin{figure}[t]
\centering
\begin{minipage}[b]{0.46\textwidth}
\centering
\begin{tikzpicture}[v/.style={circle, draw, fill=white, inner sep=1.5pt,
                              minimum size=0.55cm, font=\small}]
\foreach \i in {0,...,4} {
  \coordinate (p\i) at ({90-72*\i}:1.55);
}
\foreach \i/\j in {0/2,1/3,2/4,3/0,4/1}
  \draw[gray!60, dashed] (p\i) -- (p\j);
\foreach \i/\j in {0/1,1/2,2/3,3/4,4/0}
  \draw[thick] (p\i) -- (p\j);
\foreach \i in {0,...,4} {
  \node[v] at (p\i) {$\i$};
}
\node[font=\scriptsize, anchor=south] at ($(p0)+(0,0.3)$) {$\{e_0,e_1\}$};
\node[font=\scriptsize, anchor=west]  at ($(p1)+(0.32,0)$) {$\{e_1,e_2\}$};
\node[font=\scriptsize, anchor=north west] at ($(p2)+(0.18,-0.18)$) {$\{e_2,e_3\}$};
\node[font=\scriptsize, anchor=north east] at ($(p3)+(-0.18,-0.18)$) {$\{e_3,e_4\}$};
\node[font=\scriptsize, anchor=east]  at ($(p4)+(-0.32,0)$) {$\{e_4,e_0\}$};
\node[font=\small] at (0,-2.55) {(a) flat witness, $\frac12\le\alpha\le2$};
\end{tikzpicture}
\end{minipage}\hfill
\begin{minipage}[b]{0.5\textwidth}
\centering
\begin{tikzpicture}[x={(2.1cm,0cm)}, y={(0.8cm,1.0cm)}, z={(0cm,2.5cm)},
                    >=Latex]
\pgfmathsetmacro{\cg}{5^(-0.25)}
\pgfmathsetmacro{\sg}{sqrt(1-\cg*\cg)}
\draw[gray!55, dashed] plot[domain=0:360, samples=90, smooth]
   ({\sg*cos(\x)},{\sg*sin(\x)},{\cg});
\foreach \i in {0,...,4} {
  \pgfmathsetmacro{\a}{72*\i+57}
  \draw[->, thick, blue!60!black] (0,0,0) -- ({\sg*cos(\a)},{\sg*sin(\a)},{\cg});
  \node[font=\scriptsize, blue!50!black] at ({1.22*\sg*cos(\a)},{1.22*\sg*sin(\a)},{\cg+0.06}) {$u_{\i}$};
}
\draw[->, very thick] (0,0,0) -- (0,0,1.15) node[above, font=\small] {$c$};
\fill (0,0,0) circle (1.2pt);
\node[font=\small] at (0,0,-0.42) {(b) Lov\'asz's umbrella, $\alpha=0$};
\end{tikzpicture}
\end{minipage}
\caption{Two pentagon representations attaining the known values in the
indicated regimes.  Solid lines are edges and dashed chords are the
non-edges, which must be orthogonal.  (a) Vertex $i$ gets the uniform state
on $\{e_i,e_{i+1}\}$ and the center is $I/5$.  Non-adjacent supports are
disjoint, and the radius is $\log\frac52$ at every order; this witness is
optimal for $\frac12\le\alpha\le2$.  (b) Unit vectors $u_i$ on a cone around
the handle $c$ with $|\langle c,u_i\rangle|^2=1/\sqrt5$ and
$u_i\perp u_{i+2}$.  With the singular center $\proj c$, a limit of faithful
centers, this attains $\vt(C_5)=\sqrt5$ at order zero; for these vectors no
faithful center attains it.  For
$0\le\alpha\le\frac12$, faithful centers tending to $\proj c$ make the
exponentiated radius tend to $5^{1/(2-2\alpha)}$, giving an upper bound on
$\vt_\alpha(C_5)$.  For
$0<\alpha<\frac12$, the exact graph value remains undetermined.}
\label{fig:pentagon-geometries}
\end{figure}

\subsection{Below order one half}

Order zero is governed by a trace-norm version of \cref{lem:hs}, with the
constant $\sqrt5$ in place of $2$.  Interpolating between the two gives a
lower bound at every order in between.

\begin{lemma}[trace-norm inequality]\label[lemma]{lem:trace}
For every representation of $C_5$ and every operator $X$,
$\sum_i\nrm{P_iXP_i}_1\le\sqrt5\,\nrm X_1$.
\end{lemma}

\begin{proof}
If $(P_i)$ represents $G$ and $v$ is a vector, pick unit vectors
$u_i\in\ran P_i$ parallel to $P_iv$ (arbitrary if $P_iv=0$).  They form an
orthonormal representation of $G$, and $\sum_i\nrm{P_iv}^2=\sum_i|\langle
u_i,v\rangle|^2\le\vt(\overline G)\nrm v^2$ by Lov\'asz's Theorem~5 applied
to $\overline G$ \cite{Lovasz1979}.  That theorem is stated for real vectors;
to apply it, rephase each $u_i$ so that $\langle u_i,v\rangle\ge0$ and
realify, $w\mapsto(\operatorname{Re}w,\operatorname{Im}w)$, which preserves
unit length, orthogonality and the values $\langle u_i,v\rangle$.  The
pentagon is self-complementary, so the constant is
$\sqrt5$.  For $X=ab^*$ we have $\nrm{P_iab^*P_i}_1=\nrm{P_ia}\nrm{P_ib}$, and
Cauchy--Schwarz proves the claim for rank one.  A singular value decomposition
and the triangle inequality do the rest.
\end{proof}

\begin{lemma}[interpolation]\label[lemma]{lem:interp}
For $1\le p\le2$, every representation of $C_5$ and every operator $X$,
\[
 \sum_i\nrm{P_iXP_i}_p^p\le B(p)\,\nrm X_p^p,
 \qquad B(p)=5^{(2-p)/2}\,2^{p-1}.
\]
\end{lemma}

\begin{proof}
Consider the map $T:X\mapsto\bigoplus_iP_iXP_i$ into block-diagonal
operators on $\mathbb C^5\otimes\cH$.  For block-diagonal operators
$\nrm{\bigoplus_iY_i}_p^p=\sum_i\nrm{Y_i}_p^p$, so the claim says that $T$
has norm at most $B(p)^{1/p}$ from $S_p(\cH)$ to $S_p(\mathbb C^5\otimes\cH)$.
By \cref{lem:trace} the norm is at most $\sqrt5$ for $p=1$, and by
\cref{lem:hs} it is at most $\sqrt2$ for $p=2$.  The Schatten classes form a
complex interpolation scale with equal norms \cite[Eq.~(2.1)]{PisierXu2003},
and the complex method is exact of exponent $\theta$
\cite{BerghLofstrom1976}, here $\theta=2-2/p$.  So the norm from $S_p$ to
$S_p$ is at most
$\sqrt5^{\,1-\theta}\sqrt2^{\,\theta}=5^{(2-p)/(2p)}2^{(p-1)/p}$.  Raise this
to the power $p$.
\end{proof}

\begin{theorem}\label{thm:pentagon-bounds}
For $0\le\alpha\le\frac12$ and $p=1/(1-\alpha)$, the bounds
\eqref{eq:pentagon-intro} hold.
\end{theorem}

\begin{proof}
For the lower bound put $s=1/p$ and $X=\sigma^s$ in \cref{lem:interp}.  Then
$\nrm X_p^p=\Tr\sigma=1$ and $\nrm{P_i\sigma^sP_i}_p^p=Z_s(\sigma,P_i)$, so
$\sum_iZ_s(\sigma,P_i)\le B(p)$.  The minimum is at most the average, so
$\vt_\alpha(C_5)\ge5/B(p)=\sqrt5(\sqrt5/2)^{p-1}$.

The step from the minimum to the average loses nothing.  Given a
representation $(P_i)$ of $C_5$ with center $\sigma$, the projections
$P'_i=\bigoplus_{k\in\mathbb Z_5}P_{i+k}$ on $\mathbb C^5\otimes\cH$ again
represent $C_5$, and with the center $\sigma'=\bigoplus_k\sigma/5$
homogeneity and additivity over orthogonal blocks give
$Z_s(\sigma',P'_i)=\frac15\sum_kZ_s(\sigma,P_k)$ for every $i$.  So the
minimum for $(P'_i)$ and $\sigma'$ is the average for $(P_i)$ and $\sigma$,
and the supremum of the minimum equals the supremum of the average.

The bound $\frac52$ is \cref{thm:pentagon-half} combined with monotonicity.
For the other upper bound, take the umbrella of
\cref{fig:pentagon-geometries}(b): unit vectors
$u_i=(\cos\gamma,\sin\gamma\cos\frac{2\pi i}5,\sin\gamma\sin\frac{2\pi i}5)$
with $\cos^2\gamma=1/\sqrt5$, which satisfy $u_i\perp u_{i+2}$.  For a pure
state, $\Tr\proj{u}^\alpha\sigma^{1-\alpha}=\bra u\sigma^{1-\alpha}\ket u$.
As faithful centers approach $\proj c$ with $c=(1,0,0)$, this tends to
$|\langle c,u_i\rangle|^2=5^{-1/2}$ because $1-\alpha>0$.  The radius
therefore tends to $\frac{1}{1-\alpha}\log\sqrt5=\log5^{p/2}$.
\end{proof}

\subsection{A continuum of spectral points}

\begin{proof}[Proof of \cref{thm:B}]
The values on $[\frac12,2]$ are \cref{thm:pentagon-half}, and the bounds on
$[0,\frac12]$ are \cref{thm:pentagon-bounds}.  By
\cref{prop:continuity,prop:endpoints}, $\alpha\mapsto\vt_\alpha(C_5)$ is
continuous and nondecreasing on $[0,\frac12]$.  It runs from $\vt(C_5)=\sqrt5$
to $\frac52$, so by the intermediate value theorem its image is
$[\sqrt5,\frac52]$.  Parameters with different values on $C_5$ are different
functions, and each of them lies in $\Xs$ by \cref{thm:A}.
\end{proof}

The lower bound in~\eqref{eq:pentagon-intro} exceeds $\sqrt5$ as soon as
$\alpha>0$.  So no $\vt_\alpha$ with $\alpha>0$ equals $\vt$.  Points of the
\emph{ordinary} spectrum whose values on $C_5$ fill $[\sqrt5,\frac52]$ were
already available through Vrana's interpolation (\cref{sec:new}); the new
content of \cref{thm:B} is that the points $\vt_\alpha$ lie in $\Xs$.

\begin{remark}[an explicit ladder]
Readers who want explicit separations without the intermediate value theorem
can use $\alpha_k=1/(8^{k+1}+1)$, $k\ge0$.  Then $p_k-1=8^{-(k+1)}$, and the
bounds~\eqref{eq:pentagon-intro} give
$\vt_{\alpha_{k+1}}(C_5)\le\sqrt5\cdot5^{(p_k-1)/16}
<\sqrt5\,(\sqrt5/2)^{p_k-1}\le\vt_{\alpha_k}(C_5)$.  The strict inequality is
$5^{1/16}<\sqrt5/2$, that is, $2^{16}=65536<78125=5^7$.  Together with
$\vt_{\alpha_0}(C_5)\le5^{9/16}<\frac52$, this shows that
$\vt,\vt_{\alpha_0},\vt_{\alpha_1},\dots,\vt_{1/2}$ are pairwise distinct.
\end{remark}

\subsection{The Duan--Winter chain}\label{sec:dw-chain}

Duan and Winter prove the chain~\eqref{eq:dw-chain}
\cite[Eq.~(54)]{DuanWinter2016}.  They show that its last link can be strict,
and they conjecture that on the pentagon every term except the first equals
$\log\frac52$.  They also ask whether the two middle links can ever be strict.
We can now answer all of this.

Since the corollaries below answer their questions, we quote the relevant
passages.  The quotations, and the section, theorem, proposition and lemma
numbers we cite from \cite{DuanWinter2016}, are from arXiv:1409.3426v3,
where the chain is numbered~(53) rather than~(54).  There a
graph $G$ comes with $S_G=\operatorname{span}\{\ket i\!\bra j:i\sim j\}$,
where ``$i\sim j$ means confusability: $i=j$ or $\{i,j\}\in E$ is an edge of
the graph''.  For a cq-graph $K$ the condition $K^\dagger K<S_G$ says that
the channel is compatible with $G$ in the sense of \cref{def:representation},
and $G\times H$ denotes their strong product, our $G\boxt H$.  Section~V
defines
\begin{quote}\small
$\Sigma(G):=\inf\{\Sigma(K):K$ cq-graph with $K^\dagger K<S_G\}$, and the
asymptotic simulation cost (regularization)
$S_{0,\mathrm{NS}}(G):=\lim_{n\to\infty}\frac1n\log\Sigma(G^n)$, where
$G^n=G\times\cdots\times G$ denotes the $n$-fold strong graph product.  The
latter limit exists and equals the infimum because evidently
$\log\Sigma(G\times H)\le\log\Sigma(G)+\log\Sigma(H)$.
\end{quote}
Their Lemma~23 reads: ``For any two graphs $G$ and $H$,
$\Cmin(G\times H)=\Cmin(G)+\Cmin(H)$.''  They derive the chain from it as a
corollary and continue:
\begin{quote}\small
Note that it may be true that $S_{0,\mathrm{NS}}(G)=\log\Sigma(G)$, but to
prove this we would need to show the additivity relation
$\log\Sigma(G\times H)=\log\Sigma(G)+\log\Sigma(H)$, which remains unknown.
\end{quote}
After an example in which the last link is strict, they write:
\begin{quote}\small
We think that also the leftmost inequality in (53) can be strict.  But
although the pentagon $G=C_5$ seems to be a candidate, for which we
conjecture (based on some ad hoc calculations) that
$\Cmin(G)=S_{0,\mathrm{NS}}(G)=\log\Sigma(G)=\log\alpha^*(G)=\log\frac52$,
whereas $\log\vt(G)=\frac12\log5$, a rigorous proof of this has so far eluded
us.  Whether the other two inequalities can be strict remains an open
question.
\end{quote}

Their Proposition~18 does not settle the middle link.  It states that
$\Sigma(K_1\otimes K_2)=\Sigma(K_1)\Sigma(K_2)$ for cq-graphs $K_1,K_2$,
which concerns two fixed cq-graphs, that is, two fixed families of output
supports.  The graph quantity
$\Sigma(G\boxt H)$ is an infimum over all channels compatible with $G\boxt H$,
and most of them are not of the form $K\otimes L$.  So Proposition~18 gives
only $\Sigma(G\boxt H)\le\Sigma(G)\Sigma(H)$, the inequality quoted above.  The
reverse inequality needs a bound for non-product channels.  That is the
additivity they list as unknown, and \cref{thm:product} at order two supplies
it (\cref{rem:single-letter}).  Operationally, $\SNS(G)$ is the cost rate of
exact no-signalling-assisted simulation when, at each block length $n$, one
may choose any channel compatible with $G^{\boxt n}$.  Such a channel
may correlate the $n$ uses.  \Cref{cor:middle} says that this freedom does not
lower the rate below the single-letter value $\log\Sigma(G)$.

\begin{corollary}[the pentagon conjecture]\label[corollary]{cor:pentagon}
$\log\sqrt5<\Cmin(C_5)=\SNS(C_5)=\log\Sigma(C_5)=\log\alpha^*(C_5)=\log\frac52$.
\end{corollary}

\begin{proof}
\Cref{thm:pentagon-half,prop:endpoints} give
$\Cmin(C_5)=\log\Sigma(C_5)=\log\frac52$, and~\eqref{eq:dw-chain} squeezes
$\SNS(C_5)$ between them.  The fractional packing number of $C_5$ is
$\frac52$ \cite{Shannon1956} and $\vt(C_5)=\sqrt5$ \cite{Lovasz1979}.
\end{proof}

The values $\vt(C_5)=\sqrt5$ and $\alpha^*(C_5)=\frac52$ are classical, and
the upper bounds on $\Cmin(C_5)$, $\SNS(C_5)$ and $\log\Sigma(C_5)$ are Duan
and Winter's chain.  The new ingredient is the lower bound
$\Cmin(C_5)\ge\log\frac52$.  This proof uses only \cref{lem:hs}, the
projection formula, the radius formula for the Holevo capacity and
monotonicity in $\alpha$.  The product theorem is not needed for the pentagon,
but it is needed for the general statement.

\begin{corollary}[the middle links]\label[corollary]{cor:middle}
For every graph $G$, $\SNS(G)=\log\Sigma(G)$.  There is a graph with
$\Cmin(G)<\SNS(G)$.
\end{corollary}

\begin{proof}
The first statement is \cref{rem:single-letter}.  The second is
\cref{thm:C}, since $\Cmin(G_{32})<\log\Sigma(G_{32})=\SNS(G_{32})$.
\end{proof}

For a fixed cq-graph, strict inequalities of this kind were known: for the
cq-graph $K$ of two non-orthogonal pure states, Duan and Winter show
$C_{0,\mathrm{NS}}(K)<C_{\min E}(K)<S_{0,\mathrm{NS}}(K)$
\cite[Section~IV.D and Fig.~4]{DuanWinter2016}.  \Cref{cor:middle} separates
the graph-level quantities, each of which is an infimum over all channels
compatible with $G$.

For non-commutative bipartite graphs $K$, Duan and Winter also asked
whether $\Sigma(K)$ can be strictly submultiplicative
\cite[Section~VII]{DuanWinter2016}.  Wang and Duan answered this with
explicit qubit--qutrit examples, for which
$\Sigma(K\otimes K)<\Sigma(K)^2$ \cite{WangDuan2016}.  For the graph-level
quantity $\Sigma(G)$, with the optimization over compatible channels built
in, \cref{thm:product} shows that no such gap occurs.

\section{Order one versus order two}\label{sec:small}

\subsection{Where to look}

\Cref{lem:flat} tells us where \emph{not} to look.  If some flat witness is
optimal at order one, then orders one and two agree.  This is what happens on
the pentagon, and on every graph whose best order-one channel is flat.  To
separate the endpoints we need a graph whose good channels are genuinely
quantum.  Contextuality scenarios are a natural source, since they come with
rigid noncommutative representations.  We take the most familiar one, the
Mermin--Peres magic square \cite{Mermin1990,Peres1990}, and then add weights
by cloning vertices.  This section uses \cref{sec:compressions} but not
\cref{thm:product}.

\subsection{The graph}

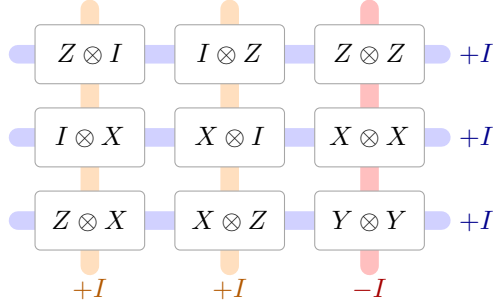
\begin{figure}[t]
\centering
\begin{tikzpicture}[cell/.style={draw=gray!70, fill=white, rounded corners=2pt,
     minimum width=1.45cm, minimum height=0.78cm, font=\small}]
\def\cx{1.85}\def\cy{1.1}
\foreach \r in {0,1,2}
  \draw[line width=7pt, blue!18, line cap=round] (-0.95,-\cy*\r) -- (2*\cx+0.95,-\cy*\r);
\foreach \c in {0,1}
  \draw[line width=7pt, orange!25, line cap=round] (\cx*\c,0.6) -- (\cx*\c,-2*\cy-0.6);
\draw[line width=7pt, red!25, line cap=round] (2*\cx,0.6) -- (2*\cx,-2*\cy-0.6);
\node[cell] at (0,0)            {$Z\otimes I$};
\node[cell] at (\cx,0)          {$I\otimes Z$};
\node[cell] at (2*\cx,0)        {$Z\otimes Z$};
\node[cell] at (0,-\cy)         {$I\otimes X$};
\node[cell] at (\cx,-\cy)       {$X\otimes I$};
\node[cell] at (2*\cx,-\cy)     {$X\otimes X$};
\node[cell] at (0,-2*\cy)       {$Z\otimes X$};
\node[cell] at (\cx,-2*\cy)     {$X\otimes Z$};
\node[cell] at (2*\cx,-2*\cy)   {$Y\otimes Y$};
\foreach \r in {0,1,2}
  \node[font=\small, blue!55!black, anchor=west] at (2*\cx+1.05,-\cy*\r) {$+I$};
\node[font=\small, orange!70!black] at (0,-2*\cy-0.9)   {$+I$};
\node[font=\small, orange!70!black] at (\cx,-2*\cy-0.9) {$+I$};
\node[font=\small, red!65!black]    at (2*\cx,-2*\cy-0.9) {$-I$};
\end{tikzpicture}
\caption{The Mermin--Peres magic square.  Each row and each column is a
\emph{context}: its three observables commute, and their product is the
operator shown at its end.  A vertex of the seed graph $F$ is a context
together with a joint eigenvalue pattern (three signs with the prescribed
product).  That gives four vertices per context and $24$ in all.}
\label{fig:magic}
\end{figure}

Index the nine cells of a $3\times3$ array by $(i,j)\in\{0,1,2\}^2$, and fill
them with the two-qubit observables of \cref{fig:magic}.  The six contexts are
the three rows and the three columns.  A vertex of the \emph{seed graph} $F$
is a context together with signs for its three cells: signs with product $+1$
on rows and on the first two columns, and with product $-1$ on the last
column.  Two distinct vertices are non-adjacent if they belong to the same
context, or if they give opposite signs to a shared cell.  That makes
$24$ vertices, each non-adjacent to exactly nine others ($216$ ordered
non-adjacent pairs).

For a context whose first two observables are $A$ and $B$, the vertex with
first two signs $(s,t)$ is represented by the joint eigenprojection
\[
 Q_{A,B;s,t}=\tfrac14(I+sA)(I+tB),
\]
a rank-one projection on $\mathbb C^2\otimes\mathbb C^2$.  These projections
represent $F$.  In fact $Q_vQ_w=0$ holds exactly for the non-adjacent pairs,
with no accidental orthogonality.  The ranges of these $24$ projections are
the $24$ vectors, in six bases of $\mathbb C^4$, that Cubitt, Leung, Matthews
and Winter took from Peres's proof of the Kochen--Specker theorem to show
that entanglement can increase the one-shot zero-error capacity of a
classical channel \cite[Fig.~2]{CLMW2010}.  The six contexts give their six
bases, and the confusability graph of their channel is the orthogonality
graph of the $24$ vectors, the complement of $F$.  Let vertex $0$ be the
first-row vertex
with signs $(+,+,+)$, so $Q_0=\proj{00}$, and sort the vertices by overlap
$j_v=4\Tr Q_vQ_0\in\{0,1,2,4\}$.  The four classes have sizes $9,8,6,1$.
Class $0$ consists of the nine non-neighbours of vertex $0$, and class $4$ is
vertex $0$ itself.

Weights enter through cloning.  For a profile $k=(k_0,k_1,k_2,k_4)$ of
nonnegative integers, the \emph{clone graph} $F[k]$ replaces every class-$j$
vertex by $k_j$ pairwise non-adjacent copies (deleting it if $k_j=0$).  Each
copy inherits its parent's adjacencies to all other vertices.  Our main
example is
\[
 G_{32}=F[(2,1,1,0)],\qquad |V(G_{32})|=2\cdot9+8+6=32 .
\]

\subsection{Two weights and a bridge}

Two of the weights of \cref{sec:compressions} matter here.  The Gibbs weight
$w_\sigma(P)=Z_0(\sigma,P)$ belongs to order one, and the harmonic weight
$h_\sigma(P)=Z_{-1}(\sigma,P)=\Tr\sigma_{/P}$ belongs to order two.  For a
profile $k$ and a seed vertex $v$, write $k_v=k_{j_v}$ for the number of
copies of $v$ in $F[k]$.  Let $t_W(F[k])$ be the supremum of
$\min_{v:\,k_v>0}w_\sigma(P_v)/k_v$ over representations $(P_v)$ of the seed
$F$ and faithful states $\sigma$.  Define $t_H(F[k])$ in the same way with
$h_\sigma$.  Deleted vertices, those with $k_v=0$, may be represented by
$P_v=0$.

\begin{proposition}[bridge]\label[proposition]{prop:bridge}
For every profile $k$,
\[
 2^{\Cmin(F[k])}\le\frac1{t_W(F[k])}
 \qquad\text{and}\qquad
 \Sigma(F[k])\ge\frac1{t_H(F[k])} .
\]
\end{proposition}

\begin{proof}
\emph{Upper bound on $\Cmin$.}  Start from a representation $(P_v)$ of $F$
and a state $\sigma$.  Let $A=\prod_{v:k_v>0}\{1,\dots,k_v\}$, work on
$\bigoplus_{f\in A}\cH$, and put
\[
 \sigma'=\frac1{|A|}\bigoplus_{f\in A}\sigma,\qquad
 P'_{v,a}=\bigoplus_{f\in A:\,f(v)=a}P_v
\]
(zero in the other summands).  Copies of one vertex live in disjoint sets of
summands, so they are orthogonal, and non-adjacent seed vertices stay
orthogonal.  So $(P'_{v,a})$ represents $F[k]$.  Functional calculus on the
direct sum gives $w_{\sigma'}(P'_{v,a})=w_\sigma(P_v)/k_v$.  Now apply the
projection formula~\eqref{eq:projection-formula} at $\alpha=1$ together with
$\vt_1=2^{\Cmin}$.

\emph{Lower bound on $\Sigma$.}  Take a compatible family $(\rho_{v,a})$ for
$F[k]$ and $T\ge\rho_{v,a}$, put $t=\Tr T$ and $\sigma=T/t$ (regularized if
needed), and let $P_{v,a}$ be the supports.  As in
\cref{prop:endpoints}, $h_\sigma(P_{v,a})\ge1/t$.  Copies of one vertex are
non-adjacent, so $P_v=\sum_aP_{v,a}$ is a projection, and $(P_v)$ (with
$P_v=0$ for deleted $v$) represents $F$.  The harmonic weight is
superadditive over orthogonal blocks: for a positive definite block matrix
$M$, each diagonal block of $M^{-1}$ dominates the inverse of the
corresponding block of $M$.  Applied to $M=P_v\sigma^{-1}P_v$ this gives
$h_\sigma(P_v)\ge\sum_ah_\sigma(P_{v,a})\ge k_v/t$.  Hence
$t_H(F[k])\ge1/t$.
\end{proof}

The upper bound on $\Cmin$ needs one good witness.  The lower bound on
$\Sigma$ needs an upper bound on $t_H$ valid for \emph{every} representation,
in every dimension.  That is the job of the next subsection.

\subsection{Harmonic weights as moments}

Let $(P_v)$ represent $F$ on $\cH$ and let $\sigma$ be faithful.  Fix an
orthonormal basis $(e_i)$ of $\cH$, write
$\vect(M)=\sum_{i,j}\langle e_i,Me_j\rangle\,e_i\otimes e_j$, and take all
transposes in this basis.  Put
$\Omega=\vect(\sqrt\sigma)\in\cH\otimes\cH$, let $F_v$ be the orthogonal
projection onto $\sigma^{-1/2}\ran P_v$, and let $E_v=F_v^{\mathsf T}$.

\begin{lemma}[harmonic weight as a moment]\label[lemma]{lem:harmonic-moment}
For every $v$,
\[
 (P_v\otimes E_v)\Omega=(I\otimes E_v)\Omega
 \qquad\text{and}\qquad
 \langle\Omega,(I\otimes E_v)\Omega\rangle=h_\sigma(P_v).
\]
\end{lemma}

\begin{proof}
With $(A\otimes B)\vect(M)=\vect(AMB^{\mathsf T})$, the first identity says
$P_v\sqrt\sigma F_v=\sqrt\sigma F_v$.  This holds because $\sqrt\sigma F_v$
has range $\ran P_v$.  For the second, $\sqrt\sigma F_v\sqrt\sigma$ lies below
$\sigma$ and has range in $\ran P_v$.  If $0\le A\le\sigma$ has range in
$\ran P_v$, then $\sigma^{-1/2}A\sigma^{-1/2}$ is a contraction supported in
$\ran F_v$, hence at most $F_v$.  So
$\sqrt\sigma F_v\sqrt\sigma=\sigma_{/P_v}$, and its trace is the moment.
\end{proof}

So harmonic weights are expectations in the vector $\Omega$ of two commuting
families of projections, $P_v\otimes I$ and $I\otimes E_w$.  The first family
obeys the orthogonality relations of $F$.  The second obeys no relations of its
own, but it is tied to the first by the identity of
\cref{lem:harmonic-moment}.  This is exactly the setting of noncommutative
moment relaxations \cite{NPA2008,Mghirbi2026}.  We form the Gram matrix of the
$409$ vectors
\[
 \Omega,\quad (P_a\otimes I)\Omega,\quad (I\otimes E_b)\Omega,\quad
 (P_a\otimes E_b)\Omega\ \ (a=b\text{ or }a\sim b).
\]
The pairs with $a\not\simeq b$ are left out because those vectors vanish:
$(P_a\otimes E_b)\Omega=(P_aP_b\otimes E_b)\Omega=0$.  The Gram matrix is
positive semidefinite, its corner entry is $\Tr\sigma=1$, and the entry
$\langle\Omega,(I\otimes E_v)\Omega\rangle$ is $h_\sigma(P_v)$.  Many entries
coincide or vanish in every model.  We record these coincidences in a table
of $1009$ moment keys, using four rewriting rules together with symmetry
(\cref{app:certificate}).  Maximizing $t$ subject to positivity, the corner
normalization and $h_v\ge tk_v$ gives an upper bound on $t_H(F[k])$ in every
dimension.  A rational dual solution then certifies a rational upper bound.

\subsection{Certified values}

\begin{table}[t]
\centering
\begin{tabular}{@{}lcrccc@{}}
\toprule
graph & profile $(k_0,k_1,k_2,k_4)$ & vertices & $2^{\Cmin}\le$ & $\Sigma\ge$ & certified gap $\ge$ (bits)\\
\midrule
$G_{32}$ & $(2,1,1,0)$ & 32 & $6.186983\ldots$ & $6.194232\ldots$ & $0.001689$\\
$G_{44}$ & $(1,2,3,1)$ & 44 & $8.949075\ldots$ & $8.962394\ldots$ & $0.002145$\\
$G_{72}$ & $(4,3,2,0)$ & 72 & $105/8$          & $13.191800\ldots$ & $0.007324$\\
\bottomrule
\end{tabular}
\caption{Certified bounds for three clone graphs.  Upper bounds on
$2^{\Cmin}$ come from explicit rational Gibbs witnesses: a direct sum of the
magic-square representation, with center proportional to
$\mathrm{diag}(A^4,B^4,B^4,B^4)$ and $(A,B)=(41,56)$, $(56,51)$, $(3,4)$, and
one-dimensional clique representations (see the text).
Lower bounds on $\Sigma$ come from rational dual certificates for the moment
relaxation.  The gap is $\log_2$ of the ratio of the two bounds, rounded
down.}
\label{tab:certified}
\end{table}

\Cref{tab:certified} lists the certified bounds.  For $G_{32}$ the certified
rational bounds, not the exact values of $t_W$ and $t_H$, are
\begin{equation}\label{eq:g32-rationals}
 t_W(G_{32})\ \ge\ \frac{5225365592240894}{32329248967670751}
 \ >\ \frac{15147154828811}{93824992214030}\ \ge\ t_H(G_{32}).
\end{equation}
For $G_{44}$ and $G_{72}$ the certified harmonic bounds are
$t_H\le31406288103814/281475531618259$ and
$t_H\le21337116511217/281474976708840$.

\begin{proof}[Proof of \cref{thm:C}]
Combine~\eqref{eq:g32-rationals} with \cref{prop:bridge}.  The reciprocals of
the two rationals, $6.186983\ldots$ and $6.194232\ldots$, are the bounds in
\cref{thm:C}.
\end{proof}

\paragraph{The Gibbs witnesses.}
Each upper bound on $2^{\Cmin}$ comes from one representation of the seed on
$\mathbb C^4\oplus\mathbb C^m$.  Vertex $v$ gets
$Q_v\oplus\bigoplus_{C\ni v}\proj C$, where $C$ runs over $m$ weighted
cliques of $F$, each with its own basis vector $\ket C$ ($m=18$, $10$ and $9$
cliques, all of size five, for $G_{32}$, $G_{44}$ and $G_{72}$).
Non-adjacent vertices share no clique, so this is a representation.  The
center is
\[
 \sigma=\frac{w}{A^4+3B^4}\,\mathrm{diag}(A^4,B^4,B^4,B^4)\ \oplus\
 \bigoplus_Cw_C\proj C
\]
in the basis $\ket{00},\ket{01},\ket{10},\ket{11}$ of $\mathbb C^4$, with
$(A,B)=(41,56)$, $(56,51)$ and $(3,4)$ and positive rational weights $w$ and
$w_C$.  If $\sigma$ is diagonal in an orthonormal basis $(\ket i)$ with
eigenvalues $\lambda_i$, the Gibbs weight of a rank-one projection $Q$ is
$Z_0(\sigma,Q)=\prod_i\lambda_i^{\bra iQ\ket i}$.  Since
$4\bra{00}Q_v\ket{00}=j_v$, the magic-square block contributes the rational
number $wA^{j_v}B^{4-j_v}/(A^4+3B^4)$ to $w_\sigma(P_v)$, each clique
containing $v$ contributes $w_C$, and Gibbs weights add over orthogonal
blocks.  By homogeneity,
$t_W(F[k])\ge\min_{k_v>0}w_\sigma(P_v)/(k_v\Tr\sigma)$, and this minimum is
the certified value, for $G_{32}$ the left rational
in~\eqref{eq:g32-rationals}.  The stored weights have $\Tr\sigma=999999999/10^9$ for
$G_{32}$ and $G_{44}$ and $\Tr\sigma=1$ for $G_{72}$; for $G_{32}$ the
magic-square block has weight $w=345767129/10^9$.

\paragraph{Where the separations sit.}  In all three graphs the Lov\'asz number
equals the independence number, and the fractional packing number is larger:
\[
 \vt=6,\ \alpha^*=\tfrac{32}5\ (G_{32});\qquad
 \vt=8,\ \alpha^*=\tfrac{37}4\ (G_{44});\qquad
 \vt=12,\ \alpha^*=\tfrac{72}5\ (G_{72}).
\]
So the certified values of \cref{tab:certified} lie strictly inside the
interval $[\vt,\alpha^*]$ that the Duan--Winter chain allows, and the gap
between orders one and two is a small part of it.  The exact positivity of
each dual matrix $Z$ is proved in integers (\cref{app:certificate}); in
floating point its smallest eigenvalue is about $2.0\times10^{-5}$,
$6.2\times10^{-7}$ and $2.0\times10^{-5}$ for $G_{32}$, $G_{44}$ and $G_{72}$.
These margins only guided the rounding.  The values of $\vt$ and $\alpha^*$
are checked exactly by the verification script
\texttt{theorem\_c\_context.py} (see Data and code availability): an
independent set and a rational dual matrix for $\vt$, and a rational packing
and clique cover of equal value for $\alpha^*$.

\section{Order two versus the projective rank}\label{sec:projective}

\subsection{The flat bound}

Let $\xi_f$ be the projective rank of Man\v{c}inska and Roberson
\cite{MancinskaRoberson2016}, the infimum of $d/r$ over
$d/r$-representations: rank-$r$ projections on $\mathbb C^d$, one for each
vertex, with adjacent vertices orthogonal.  A $d/r$-representation $(P_x)$ of
$\overline G$ is a representation of $G$ in the sense of
\cref{def:representation}, and the flat witness $\rho_x=P_x/r$,
$\sigma=I/d$ has radius $\log(d/r)$ at every order (\cref{lem:flat}).
Hence
\begin{equation}\label{eq:flat-projective}
 \vt_\alpha(G)\le\xi_f(\overline G)\qquad(0\le\alpha\le2),
\end{equation}
and in particular $\Sigma(G)\le\xi_f(\overline G)$.  This is at least as strong as
the bound $\alpha^*$ of \cref{prop:endpoints}(iii), since
$\xi_f(\overline G)\le\alpha^*(G)$ \cite{LiZuiddam2020}.  Equality
$\Sigma(G)=\xi_f(\overline G)$ holds whenever $\vt(G)=\alpha^*(G)$.  It
holds on $C_5$, where both sides equal $\frac52$ (by
\cref{thm:pentagon-half}, \eqref{eq:flat-projective} and
$\alpha^*(C_5)=\frac52$), and on the complement of the Hadamard graph $H_n$
in Duan and Winter's example, where both equal $n$
\cite[Section~V]{DuanWinter2016}: there $\Sigma=n$, and the vectors of $H_n$
give an $n/1$-representation of $H_n$.  \Cref{thm:D} shows that it can fail.

\subsection{The order-two bridge is exact}

For the clone graphs of \cref{sec:small}, the lower bound on $\Sigma$ in
\cref{prop:bridge} is an equality.

\begin{proposition}\label[proposition]{prop:bridge-exact}
Let $k$ be a profile with at least one positive entry.  If $(P_v)$ and
$\sigma$ are as in the definition of $t_H(F[k])$, then
\[
 \Sigma(F[k])\le\max_{v:\,k_v>0}\frac{k_v}{h_\sigma(P_v)} .
\]
Consequently $\Sigma(F[k])=1/t_H(F[k])$.
\end{proposition}

\begin{proof}
Let $L$ be a common multiple of the positive $k_v$.  On $\cH\otimes\mathbb
C^L$ put $\sigma'=\sigma\otimes I/L$.  For each $v$ with $k_v>0$, split the
standard basis of $\mathbb C^L$ into $k_v$ consecutive blocks of $L/k_v$
vectors, let $\Pi_{v,a}$ project onto the $a$-th block, and put
$P'_{v,a}=P_v\otimes\Pi_{v,a}$.  Copies of one vertex get orthogonal blocks,
and non-adjacent seed vertices have orthogonal $P_v$, so $(P'_{v,a})$
represents $F[k]$.  Since
$P'_{v,a}\sigma'^{-1}P'_{v,a}=P_v\sigma^{-1}P_v\otimes L\,\Pi_{v,a}$, the
shorted operator factorizes, $\sigma'_{/P'_{v,a}}=\sigma_{/P_v}\otimes
\Pi_{v,a}/L$, and
\[
 h_{\sigma'}(P'_{v,a})=h_\sigma(P_v)\,\frac{\Tr\Pi_{v,a}}L=\frac{h_\sigma(P_v)}{k_v}.
\]
The projection formula~\eqref{eq:projection-formula} at $\alpha=2$, with
$\vt_2=\Sigma$ (\cref{prop:endpoints}), gives the displayed bound.  Taking
the supremum over $(P_v)$ and $\sigma$ gives $\Sigma(F[k])\le1/t_H(F[k])$,
and \cref{prop:bridge} gives the reverse inequality.
\end{proof}

\subsection{A rank inequality for the magic square}

In this subsection $\tau=\Tr/d$ is the normalized trace on $\mathbb C^d$.  We
use three facts about rank: it is subadditive, $\rk(XY)\le\min(\rk X,\rk Y)$,
and multiplying by an invertible operator does not change it.  Number the six
contexts of the seed $F$ as $c=0,\dots,5$: the rows are $0,1,2$ from the top
and the columns are $3,4,5$ from the left (\cref{fig:magic}).  As before,
vertex $0$ is the row-$0$ vertex with signs $(+,+,+)$.

\begin{lemma}[anticommutators bound traces]\label[lemma]{lem:anticomm}
If $X$ and $Y$ are Hermitian unitaries on $\mathbb C^d$, then
$\lvert\Tr X\rvert\le\rk(XY+YX)$.
\end{lemma}

\begin{proof}
The operator $X+YXY=(XY+YX)Y$ has the rank of $XY+YX$, norm at most $2$ and
trace $2\Tr X$.  The trace of an operator is at most its rank times its norm
in absolute value.
\end{proof}

\begin{lemma}[rank inequality]\label[lemma]{lem:rank-ineq}
Let $(P_v)$ be projections on $\mathbb C^d$, one for each vertex of $F$, with
$P_vP_w=0$ whenever $v\ne w$ are non-adjacent; zero projections are allowed.
For each context $c$ put $S_c=\sum_{v\in c}P_v$.  Then
\begin{equation}\label{eq:rank-ineq}
 4\tau(P_0)+7\tau(S_0)+11\tau(S_1)+9\tau(S_2)
 +11\tau(S_3)+11\tau(S_4)+9\tau(S_5)\le59,
\end{equation}
with equality for the magic-square representation $(Q_v)$.
\end{lemma}

The four vertices of a context are pairwise non-adjacent, so each $S_c$ is a
projection, and $\delta_c=1-\tau(S_c)$ is the normalized rank of its defect
$I-S_c$.  The coefficients of the contexts sum to $58$, so
\eqref{eq:rank-ineq} says
\[
 4\tau(P_0)\le1+7\delta_0+11\delta_1+9\delta_2+11\delta_3+11\delta_4+9\delta_5 .
\]
If every context is complete, this is $\tau(P_0)\le\frac14$: the vertex-$0$
outcome takes at most a quarter of the space.  The lemma makes this robust,
with a penalty linear in the defects.  For $(Q_v)$ every context is complete
and $\tau(Q_0)=\frac14$, which gives equality.

\begin{proof}
\emph{Completion.}  In each context, add the defect $I-S_c$ to one of its
four projections; in row $0$, add it to $P_0$.  Each context now carries four
orthogonal projections summing to $I$, labelled by its sign patterns.  For a
cell $(i,j)$, let $U_{ij}$ be the sum of the completed projections of row
$i$, each weighted by the sign it gives to the cell, and define $V_{ij}$ in
the same way from column $j$.  These are Hermitian unitaries.  The three
$U_{ij}$ in a row commute and multiply to $I$.  The three $V_{ij}$ in a
column commute and multiply to $I$, except in the last column, where they
multiply to $-I$.

\emph{Rows and columns agree off the defects.}  Fix a cell $(i,j)$.  Let
$U_\pm$ be the sum of the original projections of row $i$ that give the cell
the sign $\pm$, and $V_\pm$ the same for column $j$.  Two vertices that give
a shared cell opposite signs are non-adjacent, so $U_+V_-=U_-V_+=0$.  Let $z$
lie in the kernels of both defects, $I-S_i$ and $I-S_{3+j}$.  Then
$U_+z+U_-z=z=V_+z+V_-z$.  Put $x=U_+z$ and $y=V_+z$.  Then
$\langle x,z-y\rangle=\langle z,U_+V_-z\rangle=0$ and
$\langle z-x,y\rangle=\langle z,U_-V_+z\rangle=0$, so
\[
 \nrm x^2=\langle x,z\rangle=\langle x,y\rangle=\langle z,y\rangle=\nrm y^2 .
\]
Hence $\nrm{x-y}^2=0$ and $U_{ij}z=2x-z=2y-z=V_{ij}z$.  The two kernels meet
in a subspace of codimension at most $d(\delta_i+\delta_{3+j})$, so
\[
 \eta_{ij}:=\rk(U_{ij}-V_{ij})/d\le\delta_i+\delta_{3+j} .
\]
Put $e_j=\sum_i\eta_{ij}$ and $\eps=e_0+e_1+e_2$.  Summing over the nine
cells gives $\eps\le3\delta$, where $\delta=\sum_c\delta_c$.

\emph{Relations.}  Write $A,B$ for $U_{00},U_{01}$, write $D,E$ for
$U_{10},U_{11}$, and write $H,J$ for $U_{20},U_{21}$.  The $U$'s in a row
commute and multiply to $I$, so the last cells of the three rows carry $AB$,
$DE$ and $HJ$.  In column $0$, $V_{20}=V_{00}V_{10}=V_{10}V_{00}$.  Replacing
each $V_{i0}$ by $U_{i0}$ costs at most $\eta_{i0}$ in normalized rank, so
\[
 \rk(H-AD)/d,\ \ \rk(H-DA)/d\ \le\ e_0,
\]
and in the same way $\rk(J-BE)/d$ and $\rk(J-EB)/d$ are at most $e_1$.  In
the last column $V_{22}=-V_{02}V_{12}=-V_{12}V_{02}$, so
\[
 \rk(HJ+AB\,DE)/d,\ \ \rk(HJ+DE\,AB)/d\ \le\ e_2 .
\]
Since $HJ-AD\,BE=(H-AD)J+AD(J-BE)$, comparing $HJ$ with $AD\,BE$ and with
$-AB\,DE$ gives $\rk(ADBE+ABDE)\le d\eps$.  Now $ADBE+ABDE=A(DB+BD)E$, and
$A$ and $E$ are unitary, so $\rk(BD+DB)\le d\eps$.  In the same way,
$DAEB+DEAB=D(AE+EA)B$ gives $\rk(AE+EA)\le d\eps$.  Finally, $A$ and $D$ lie
in column $0$, where $V_{00}$ and $V_{10}$ commute, so
$\rk(AD-DA)\le2d(\eta_{00}+\eta_{10})\le du_0$ with
$u_0=2(\delta_0+\delta_1+2\delta_3)$.  Likewise $\rk(BE-EB)\le du_1$ with
$u_1=2(\delta_0+\delta_1+2\delta_4)$.

\emph{Traces.}  By \cref{lem:anticomm}, $\lvert\tau(B)\rvert\le\eps$ and
$\lvert\tau(A)\rvert\le\eps$.  Since $A$ and $B$ commute, $AB$ is a Hermitian
unitary.  The identities $\{AB,D\}=A\{B,D\}+[D,A]B$ and
$\{AB,E\}=A[B,E]+\{A,E\}B$, together with \cref{lem:anticomm}, give
\[
 \lvert\tau(AB)\rvert\le\eps+\min(u_0,u_1)\le\eps+2(\delta_0+\delta_1+\delta_3+\delta_4).
\]
The completed row-$0$ projection with signs $(+,+,+)$ is
$\frac14(I+A)(I+B)$, and it equals $P_0+(I-S_0)$.  Taking normalized traces,
\[
 4\tau(P_0)+4\delta_0=1+\tau(A)+\tau(B)+\tau(AB)
 \le1+3\eps+2(\delta_0+\delta_1+\delta_3+\delta_4).
\]
With $\eps\le3\delta$ this is
$4\tau(P_0)\le1+7\delta_0+11\delta_1+9\delta_2+11\delta_3+11\delta_4+9\delta_5$,
which is~\eqref{eq:rank-ineq}.
\end{proof}

The relations used in the proof are robust forms of the anticommutation
relations of the magic square: in an exact operator solution, cell
observables in different rows and columns anticommute
\cite[Lemma~6]{Ji2013}.

\subsection{Proof of \texorpdfstring{\cref{thm:D}}{Theorem D}}

Let $G_{1440}=F[(59,60,61,63)]$, so $k_0=59$, $k_1=60$, $k_2=61$ and
$k_4=63$; it has $9\cdot59+8\cdot60+6\cdot61+63=1440$ vertices.  In the
seed, row $0$ consists of vertex $0$ and three class-$0$ vertices, rows $1$
and $2$ consist of class-$1$ vertices, and each column consists of two
class-$0$ and two class-$2$ vertices.  So every context of $G_{1440}$ has
$63+3\cdot59=4\cdot60=2\cdot61+2\cdot59=240$ vertices.

\emph{Lower bound.}  Let $(R_{v,a})$ be a $d/r$-representation of
$\overline{G_{1440}}$, where $(v,a)$ is the $a$-th copy of the seed vertex
$v$.  Distinct copies of one vertex are non-adjacent in $G_{1440}$, so
$P_v=\sum_aR_{v,a}$ is a projection of rank $k_vr$.  If $v\ne w$ are
non-adjacent in $F$, all their copies are non-adjacent, so $P_vP_w=0$.
\Cref{lem:rank-ineq} applies with $\tau(P_0)=63r/d$ and $\tau(S_c)=240r/d$
for every context, and gives
\[
 \frac rd\,(4\cdot63+58\cdot240)=\frac rd\cdot14172\le59 .
\]
So $d/r\ge14172/59$ for every $d/r$-representation, and
$\xi_f(\overline{G_{1440}})\ge14172/59$.

\emph{Upper bound.}  Take the magic-square representation $(Q_v)$ and, in
the basis $\ket{00},\ket{01},\ket{10},\ket{11}$, the faithful state
$\sigma=\mathrm{diag}(31,29,29,29)/118$.  For a rank-one projection
$Q=\proj\psi$ we have $h_\sigma(Q)=1/\langle\psi,\sigma^{-1}\psi\rangle$.
Since $\lvert\langle00|\psi_v\rangle\rvert^2=j_v/4$ for a unit vector
$\psi_v$ in $\ran Q_v$,
\[
 \frac1{h_\sigma(Q_v)}=118\Bigl(\frac{j_v}{4\cdot31}+\frac{4-j_v}{4\cdot29}\Bigr)
 =\frac{59\,(62-j_v)}{899}.
\]
Hence $k_v/h_\sigma(Q_v)=59\,k_v(62-j_v)/899$, and $k_j(62-j)$ equals
$3658$, $3660$, $3660$ and $3654$ for $j=0,1,2,4$.  By
\cref{prop:bridge-exact},
\[
 \Sigma(G_{1440})\le\frac{59\cdot3660}{899}=\frac{215940}{899}.
\]
Explicitly, with $L=\operatorname{lcm}(59,60,61,63)$ and $\Pi_{v,a}$ as in
the proof of \cref{prop:bridge-exact}, the states
$\rho_{v,a}=(k_v/L)\,Q_v\otimes\Pi_{v,a}$ are compatible with $G_{1440}$, and
all of them lie below $T=(215940/899)\,\sigma\otimes I/L$, whose trace is
$215940/899$.

\emph{Separation.}  Finally, $14172/59-215940/899=168/53041>0$.  This proves
\cref{thm:D}.  The script \texttt{projective\_rank\_separation.py} re-checks
the class structure of the contexts, the harmonic weights and this
arithmetic in exact rational arithmetic (see Data and code availability).

\paragraph{Where the profile comes from.}
Both bounds are linear in the profile.  For every profile $k$ with positive
entries, the same argument gives
\[
 \xi_f\bigl(\overline{F[k]}\bigr)\ge\frac{83k_0+80k_1+62k_2+11k_4}{59},
 \qquad
 \Sigma(F[k])\le\max_{j}\frac{k_j}{h_j(a)},
\]
where $h_j(a)$ is the harmonic weight of a class-$j$ projection $Q_v$ at the
center $\mathrm{diag}(a,b,b,b)$, $a+3b=1$.  At the central state every
harmonic weight is $\frac14$, and for profiles proportional to the weights
the two bounds coincide.  Tilting the center toward $\ket{00}$, with
$a=\frac14+3\epsilon$ and $b=\frac14-\epsilon$, changes the weights of the
classes $j=0,1,2,4$ by $-\epsilon$, $0$, $\epsilon$ and $3\epsilon$ to first
order.  For profiles proportional to the tilted weights, the lower bound then
exceeds the upper bound by a relative $12\epsilon/59$ to first order.  So the
separation appears at first order in the tilt.  The gain is small, and the
clone counts must follow the weights closely.  For $G_{1440}$ the relative
gap is $1.3\times10^{-5}$; its profile is the smallest we found.

\section{Open problems}\label{sec:open}

\paragraph{Is $\vt$ the least point?}
Every $\vt_\alpha$ dominates $\vt$, so this paper leaves the conjecture
$\Theta_*=\vt$ where it was.  A proof that $\vt$ is the minimum of $\Xs$ would
settle it.  So would a single spectral point below $\vt$ on some graph, in the
opposite direction.

\paragraph{Above order two.}
\Cref{rem:above-two} shows that our monotonicity argument breaks down above
order two.  Is some suitably defined $\vt_\alpha$ with $\alpha>2$ still in
$\Xs$?  Flat witnesses show that any such extension is at most $\alpha^*$.
A natural guess for the $\alpha\to\infty$ endpoint is the complement of the
projective rank.  Flat representations give the corresponding upper bound, but
we are missing a dimension-free argument for the other direction.  By
\cref{thm:D}, such an extension would have to keep growing beyond order two on
$G_{1440}$.

\paragraph{The projective rank.}
\Cref{thm:D} shows that $\Sigma(G)$ and $\xi_f(\overline G)$ differ, so
\cref{thm:A} does not place $\xi_f(\overline G)$ in $\Xs$.  Whether it lies in
$\Xs$ remains open \cite[after Proposition~24]{LiZuiddam2020}.  For
$G_{1440}$, each context is an independent set of $240$ vertices, so
$240\le\Sigma(G_{1440})\le\frac{215940}{899}$, while
$\frac{14172}{59}\le\xi_f(\overline{G_{1440}})\le\alpha^*(G_{1440})$.  What are
the exact values?  Is \cref{lem:rank-ineq} the best linear inequality of its
kind for the seed, and do other contextuality scenarios give larger gaps?
How small can a graph with $\Sigma(G)<\xi_f(\overline G)$ be?

\paragraph{Other divergences.}
The proof of \cref{thm:A} uses three features of the Petz family.  These are
a closed form for the supported minimum with the tower property, the
block structure used for unions and flagged mixtures, and data processing.
For sandwiched divergences the answer is complete on the full
data-processing range $[\frac12,\infty]$: \cref{prop:sandwiched} gives the
supported minimum $-\log Z_{(1-\beta)/\beta}$ and identifies the graph
radii with the Petz family on $[0,2]$.  In particular, the orders
$\frac12\le\beta<1$ give precisely the Petz sub-one family, not additional
spectral points.  Do other quantum divergences with these three features
yield spectral points outside this common family?

\paragraph{Convexity.}
In the ordinary spectrum, Vrana's probabilistic refinement makes the spectrum
log-convex \cite{Vrana2021}.  If the same held for $\le_*$, interpolating
between $\vt$ and $\Sigma$ would produce another continuum.  The following is
a heuristic, conditional on such a theory for $\le_*$; it suggests that at
least part of the R\'enyi family is not of this form.  In the ordinary
spectrum, vertex-transitive graphs are evaluated log-affinely along such
interpolations \cite[Section~5]{Vrana2022}.  If the same were true for
$\le_*$, an interpolant with value $\frac52$ on $C_5$ would carry no weight on
$\vt$, so it would be $\Sigma$ itself.  Since $\vt_1(C_5)=\frac52$ and
$\vt_1\ne\Sigma$ by \cref{thm:C}, $\vt_1$ would then not be such an
interpolant; we do not know whether the $\vt_\alpha$ with $\alpha<\frac12$
would be.  Making any of this rigorous would need an entanglement-assisted
theory of probabilistic refinements, which we do not have.

\paragraph{The pentagon curve.}
What is $\vt_\alpha(C_5)$ for $0<\alpha<\frac12$?  The bounds in
\cref{fig:pentagon-curve} do not meet.  The flat witness in
\cref{fig:pentagon-geometries}(a) is optimal on $[\frac12,2]$, and the
umbrella in panel (b) attains the order-zero value.  At intermediate orders,
the gap between the displayed bounds does not by itself decide which
representation or center is optimal; the exact graph value remains open.

\paragraph{Small separations.}
How small can a graph with $\Cmin<\log\Sigma$ be?  Thirty-two vertices
suffice, and we have no reason to think this is minimal.  Can the rational
certificate behind \cref{thm:C} be replaced by a short analytic inequality,
in the spirit of \cref{lem:hs}?  And is the first level of the moment
relaxation already exact for the three graphs of \cref{tab:certified}?

\paragraph{Non-commutative graphs.}
Duan, Severini and Winter's quantum Lov\'asz number extends $\vt$ to
non-commutative graphs \cite{DSW2013}, and Li and Zuiddam place it in the
entanglement-assisted spectrum of non-commutative graphs.  There the analogue
of $\Theta_*=\vt$ fails: Wang and Duan exhibit non-commutative graphs whose
entanglement-assisted zero-error capacity is strictly below the quantum
Lov\'asz number \cite{WangDuan2018}, so that spectrum contains another point,
although no explicit one is known \cite[Theorem~33 and the remark after
it]{LiZuiddam2020}.  Does the whole family $\vt_\alpha$
extend, and does the product theorem survive?  A positive answer at order two
would give a single-letter formula for graph-level no-signalling simulation
over non-commutative graphs, in contrast with the fixed-$K$ quantity
$\Sigma(K)$, which can be strictly submultiplicative \cite{WangDuan2016}.

\section*{Data and code availability}

The certificate data behind \cref{thm:C} and \cref{tab:certified}, the exact
data for the values of $\vt$ and $\alpha^*$ quoted in \cref{sec:small}, and
the verification scripts named in the text are archived on Zenodo under the
all-versions DOI
\href{https://doi.org/10.5281/zenodo.22983498}{10.5281/zenodo.22983498},
and are also available at
\url{https://github.com/Apsiape/the-lovasz-number-is-not-alone}.  They
include an exact check of the arithmetic behind \cref{thm:D}.  The
archive documents the data format and pins the expected output of every
script.

\section*{Reproducibility}

The certificate data and their verification scripts accompany this paper
(see Data and code availability).
They are needed only for the finite rational claims of \cref{sec:small}.
The proof of \cref{thm:D} is complete by hand; the script
\texttt{projective\_rank\_separation.py} re-checks its arithmetic.
An additional deterministic matrix regression checks instances of the
supported sandwiched minimum in \cref{prop:sandwiched}, including the
sub-one orders and endpoints.  These finite numerical checks are not used
as proof; the dimension-independent statements are proved in the text.

\section*{AI-assisted tools}

Generative AI tools were used throughout the research and preparation of this
manuscript, including for mathematical exploration, proof development and
critique, literature organization, computational assistance, and drafting.
The authors determined which results and arguments to include and take
responsibility for the final manuscript.  Exact computer-assisted claims are
supported, where indicated, by accompanying deterministic certificates and
verification scripts.

\appendix

\section{The certificate}\label[appendix]{app:certificate}

\paragraph{Words and moments.}
For words $L=a_1\cdots a_m$ and $R=b_1\cdots b_n$ in seed vertices, write
$P_L=P_{a_1}\cdots P_{a_m}$, $E_R=E_{b_1}\cdots E_{b_n}$ and
\[
 m(L,R)=\langle\Omega,(P_L\otimes E_R)\Omega\rangle .
\]
Every entry of the Gram matrix of \cref{sec:small} is such a moment.  The
following rules hold in every model.
\begin{enumerate}[label=(R\arabic*),leftmargin=2.6em,itemsep=1pt]
\item Equal adjacent letters in $L$ or in $R$ merge, since $P_a^2=P_a$ and
$E_b^2=E_b$.
\item If $L$ contains two adjacent letters $a\ne b$ that are non-adjacent in
$F$, then $m(L,R)=0$, since $P_aP_b=0$.
\item If $R$ ends in $y$, then $m(L,R)=m(Ly,R)$, since
$(I\otimes E_y)\Omega=(P_y\otimes E_y)\Omega$.
\item If $R$ begins with $y$, then $m(L,R)=m(yL,R)$, by the adjoint of the
same identity.
\end{enumerate}
Reversing both words conjugates the moment, since
$(P_L\otimes E_R)^*=P_{\bar L}\otimes E_{\bar R}$.  Let $\mathcal A_0$ be the
$48$ automorphisms of $F$ that fix vertex $0$ (equivalently, that preserve
the overlap classes).  A relabelled representation is again a representation,
and the constraints $h_v\ge tk_v$ are invariant because $k$ is constant on
classes.  So we may average the Gram matrix over $\mathcal A_0$ and then take
real parts.  Both operations preserve positivity, and together they identify
symmetry-related moments and reversed pairs.  The verifier builds the key
table from these rules only.  So every identification in the table holds for
the Gram matrix of every representation, in every dimension.  The result has
$1009$ keys.

\paragraph{Deleted vertices.}
In $G_{32}$ and $G_{72}$ the class-$4$ vertex is deleted.  It stays in the
seed and is represented by $P_0=0$, and hence $E_0=0$.  All four rules remain
valid for zero projections, no radial constraint is imposed on deleted
vertices, and the relaxation needs no change.

\paragraph{The relaxation and its dual.}
Let $\Gamma(y)$ be the real symmetric matrix with one variable $y_\kappa$ per
key, write $1$ for the key of $m(\varnothing,\varnothing)=\Tr\sigma$, and let
$\kappa_v$ be the key of $m(\varnothing,v)=h_\sigma(P_v)$.
Entries of the Gram matrix whose moment vanishes identically by the rules
(R1)--(R4) are $0$ in $\Gamma(y)$; they carry no key and no equation below,
so the corresponding entries of $Z$ are restricted only by $Z\succeq0$.  The
primal problem is to maximize $t$ subject to
\[
 \Gamma(y)\succeq0,\qquad y_1=1,\qquad y_{\kappa_v}\ge tk_v\ \ (k_v>0).
\]
A dual certificate consists of $Z\succeq0$, multipliers $\mu_v\ge0$ and a
number $\lambda$ such that, for every key $\kappa$,
\begin{equation}\label{eq:moment-dual}
 \sum_{(i,j):\,\mathrm{key}(i,j)=\kappa}Z_{ij}
 +\sum_{v:\,\kappa_v=\kappa}\mu_v=\lambda\,[\kappa=1].
\end{equation}
Pairing with $\Gamma(y)$ gives
$0\le\langle Z,\Gamma(y)\rangle=\lambda-\sum_v\mu_vy_{\kappa_v}$, and
therefore
\[
 t_H(F[k])\le\frac{\lambda}{\sum_v\mu_vk_v}.
\]

\paragraph{Exact positivity.}
The certificate file (see Data and code availability) contains, for each
graph, the rational Gibbs witness, the multipliers $\mu$, the nonzero entries
of $Z$, one for each orbit of $\mathcal A_0$ together with transposition, and
a shift $\gamma$.  The verifier rebuilds the seed, its $216$ ordered non-adjacent
pairs, the two-qubit representation, the overlap classes, the $48$
automorphisms, the $409$ monomials and the key table.  It checks
every equation~\eqref{eq:moment-dual} and $\mu\ge0$ in rational arithmetic,
and it recomputes the Gibbs witness exactly.  To prove $Z\succeq0$, it clears
denominators, $DZ\in\mathbb Z^{409\times409}$, rounds a floating-point
Cholesky factor of the shifted matrix $Z-\gamma I$ at scale $2^{40}$ to an
integer matrix $L$, and checks in exact integer arithmetic that
\[
 R=2^{80}DZ-DLL^{\mathsf T}
\]
is symmetric and diagonally dominant with nonnegative diagonal.  Then
$R\succeq0$ and $Z=LL^{\mathsf T}/2^{80}+R/(2^{80}D)\succeq0$.  The shift,
$\gamma=10^{-5}$, $2\cdot10^{-7}$ and $10^{-5}$ for $G_{32}$, $G_{44}$ and
$G_{72}$, leaves room on the diagonal of $R$ for the rounding error.  Only
the integer check counts.  Whether the floating-point factorization
succeeded plays no role.

\paragraph{Negative controls.}
The verifier runs two tampered certificates, and both must be rejected.  One
breaks an affine equation.  The other changes only entries attached to moments
that vanish identically, so every affine equation still holds but exact
positivity fails.

\paragraph{How the certificate was found.}
We followed the recovery procedure of \cite{Mghirbi2026}.  A numerical SDP
solution is symmetrized, rounded to dyadic rationals, and projected exactly
onto the affine constraints by a minimum-norm correction, and positivity is
then certified as above.  None of this matters for correctness, which rests
only on the exact checks.

\section{Deferred proofs}\label[appendix]{app:proofs}

\subsection{Continuity in the order}\label[subappendix]{app:continuity}

\begin{proof}[Proof of \cref{prop:continuity}]
Write $f(\alpha)=\log\vt_\alpha(G)$.

\emph{Right continuity at $\alpha_0\in[0,2)$.}  Pick a compatible family and
a faithful center whose order-$\alpha_0$ radius is at most $f(\alpha_0)+\eps$.
For fixed states with a faithful center, $\alpha\mapsto D_\alpha$ is
continuous on $[0,2]$, including from the right at $0$, because
$\rho^\alpha\to\Pi_\rho$, and at $1$, where the limit is the relative
entropy.  So
$\limsup_{\alpha\downarrow\alpha_0}f(\alpha)\le f(\alpha_0)+\eps$, and
monotonicity gives the reverse inequality.

\emph{Left continuity at $\alpha\in(0,1)$.}  Let $0<\beta<\alpha$ and
$\psi(t)=\ln\Tr\rho^t\sigma^{1-t}$.  By H\"older's inequality $\psi$ is
convex, and $\psi(1)=0\ge\psi(0)$.  So
$\psi(\beta)\le(1-\frac\beta\alpha)\psi(0)+\frac\beta\alpha\psi(\alpha)
\le\frac\beta\alpha\psi(\alpha)$.  Since both orders lie below one, this
rearranges to
\[
 D_\alpha(\rho\Vert\sigma)\le\frac{\alpha(1-\beta)}{\beta(1-\alpha)}
 D_\beta(\rho\Vert\sigma).
\]
The constant does not depend on the states or the dimension.  Applying it to
near-optimal witnesses at order $\beta$ gives
$f(\beta)\le f(\alpha)\le\frac{\alpha(1-\beta)}{\beta(1-\alpha)}f(\beta)$,
and the constant tends to $1$ as $\beta\uparrow\alpha$.
\end{proof}

\subsection{Orthogonal blocks}\label[subappendix]{app:orthogonal-sum}

\begin{proof}[Proof of \cref{lem:orthogonal-sum}]
For $s\ne0$ let $A=\sigma^s$ and $f(t)=t^{1/s}$.  For $s=0$ let $A=\ln\sigma$
and $f=\exp$.  In each case $f$ is convex on an interval containing the
spectrum of $A$ and
$f(A)=\sigma$.  Diagonalize $Q_jAQ_j$ on $\ran Q_j$ in an orthonormal basis
$(e_{jk})_k$.  Its eigenvalues are $\langle e_{jk},Ae_{jk}\rangle$, so
Jensen's inequality for the spectral measure of $A$ gives
\[
 Z_s(\sigma,Q_j)=\sum_kf(\langle e_{jk},Ae_{jk}\rangle)
 \le\sum_k\langle e_{jk},\sigma e_{jk}\rangle .
\]
The vectors $(e_{jk})_{j,k}$ are orthonormal, so the total is at most
$\Tr\sigma=1$.
\end{proof}

\subsection{Exact confusability graphs}\label[subappendix]{app:exact}

\begin{lemma}\label[lemma]{lem:exact-graph}
In~\eqref{eq:def-parameter}, \eqref{eq:cmin-def} and~\eqref{eq:sigma-def},
restricting to channels whose confusability graph is exactly $G$ does not
change the infimum.
\end{lemma}

\begin{proof}
Restricting can only increase an infimum, so we show the reverse.  Take a
compatible family $(\rho_x)$ on $\cH$ with center $\sigma$ (or dominator
$T$).  If $E(G)=\varnothing$, every compatible family already has exact
confusability graph $G$, so there is nothing to prove.  Assume henceforth that
$|E(G)|>0$.  Add an orthonormal vector $f_e$ for each edge $e$ of $G$.  For a
vertex $x$ of degree $d_x>0$, let $\Pi_x$ project onto the $f_e$ with $e\ni x$,
and let $\Pi_E$ project onto all of them.  Put
\[
 \rho'_x=(1-\eps)\rho_x\oplus\frac{\eps}{d_x}\Pi_x,\qquad
 \sigma'=(1-\eps)\sigma\oplus\frac{\eps}{|E|}\Pi_E,\qquad
 T'=T\oplus\eps\,\Pi_E ,
\]
and leave isolated vertices unchanged.  Adjacent $x,y$ now overlap on
$f_{\{x,y\}}$, while non-adjacent vertices share no new vector.  So the
confusability graph is exactly $G$.  Every Petz moment, relative entropy and
support weight splits over the two summands, and the new summand contributes
$O(\eps)$.  So radii converge to the original ones as $\eps\to0$.  Also
$T'\ge\rho'_x$ and $\Tr T'\to\Tr T$.  The statement for $\Cmin$ follows from
the order-one radius identity.
\end{proof}

\subsection{The entanglement-assisted preorder}\label[subappendix]{app:preorder}

Li and Zuiddam state the preorder for graphs through complements.  Their
Definition~7 \cite{LiZuiddam2020} (Definition~6 in arXiv:1810.00744v2)
writes $G\overset{*}{\to}H$ if there are
$d\in\mathbb N$ and $d\times d$ positive semidefinite matrices $\rho$ and
$\rho_g^h$ ($g\in V(G)$, $h\in V(H)$) such that
\begin{enumerate}[label=(\roman*),leftmargin=2.2em,itemsep=1pt,topsep=3pt]
\item $\sum_{h\in V(H)}\rho_g^h=\rho$ for every $g\in V(G)$, and
\item $\rho_g^h\rho_{g'}^{h'}=0$ whenever $\{g,g'\}\in E(G)$ and
$\{h,h'\}\notin E(H)$.
\end{enumerate}
The entanglement-assisted cohomomorphism preorder is then $G\le_*H$ if and
only if $\overline G\overset{*}{\to}\overline H$.  Their graphs are
confusability graphs, as ours are.

\begin{lemma}\label[lemma]{lem:preorder}
For graphs $G$ and $H$, the relation $G\le_*H$ of Li and Zuiddam, that is
$\overline G\overset{*}{\to}\overline H$, holds if and only if there are a
state $\omega$ and positive semidefinite operators $\omega_g^h$
satisfying~\eqref{eq:ea-sum} and~\eqref{eq:ea-orth}.
\end{lemma}

\begin{proof}
Apply the definition to $\overline G$ and $\overline H$.  A pair $\{g,g'\}$ is
an edge of $\overline G$ exactly when $g\ne g'$ and $g,g'$ are non-adjacent in
$G$.  A pair $\{h,h'\}$ is a non-edge of $\overline H$ exactly when $h=h'$ or
$h\sim h'$ in $H$, that is, when $h\simeq h'$.  The case $h=h'$ is included
because graphs have no loops.  So condition~(ii) for $\overline G$ and
$\overline H$ is~\eqref{eq:ea-orth}, and condition~(i) is~\eqref{eq:ea-sum}
with $\omega=\rho$.  Both conditions are homogeneous, so dividing $\rho$ and
every $\rho_g^h$ by $\Tr\rho$ makes $\rho$ a state.  (A nonzero $\rho$ is
implicit in the definition, since $\rho=0$ would relate every pair of graphs;
Li and Zuiddam note that $\rho$ may even be taken positive definite.  In
Cubitt et al.'s Definition~1, from which the definition is taken, $\rho$ is
the reduced state of a bipartite pure state, so it is a nonzero state
\cite[Definition~1]{CMRSSW2014}.)
Conversely, a state $\omega$ and operators $\omega_g^h$
satisfying~\eqref{eq:ea-sum}--\eqref{eq:ea-orth} satisfy (i) and (ii) with
$\rho=\omega$ in dimension $d=\dim\cH$.
\end{proof}

\bibliographystyle{alpha}
\bibliography{references}

\end{document}